\documentclass[12pt,draftclsnofoot,journal,onecolumn]{IEEEtran}
\IEEEoverridecommandlockouts

\usepackage{geometry}
\usepackage[noend]{algpseudocode}
\usepackage{balance}
\usepackage{algorithmicx,algorithm}
\usepackage{eqparbox}
\usepackage{amsmath}
\usepackage{subfigure}
\usepackage{verbatim}
\usepackage{braket}
\usepackage{cases}
\usepackage{dsfont}
\usepackage{amsthm,amsmath,amssymb,bm,bbm}

\newtheorem{lemma}{Lemma}
\newtheorem{theorem}{Theorem}
\newtheorem{definition}{Definition}
\newtheorem{remark}{Remark}

\usepackage{mathrsfs}

\usepackage{hyperref}
\hypersetup{hidelinks, 
	colorlinks=true,
	allcolors=black,
	pdfstartview=Fit,
	breaklinks=true}

\usepackage{graphicx}
\usepackage{epstopdf}
\usepackage{cite}
\makeatletter

\newcommand{\biggg}{\bBigg@{3}}

\newcommand{\Biggg}{\bBigg@{3.5}}

\newcommand{\bigggg}{\bBigg@{4}}

\newcommand{\Bigggg}{\bBigg@{4.5}}

\makeatother
\begin{document}
	\title{A General DoF and Pattern Analyzing Scheme for Electromagnetic Information Theory}
	\author{
		%		\vspace{0.2cm}
		\IEEEauthorblockN{
			Zhongzhichao~Wan,
			Jieao~Zhu,
			Yongli~Yan,
			and~Linglong~Dai,~\IEEEmembership{Fellow,~IEEE}
		}
%		\IEEEauthorblockN{
%			Zhongzhichao~Wan,
%			Jieao~Zhu,
%			Yongli~Yan,
%			Linglong~Dai,~\IEEEmembership{Fellow,~IEEE},
%			Husheng~Li,~\IEEEmembership{Senior Member,~IEEE},
%			and~Lizhong~Zheng,~\IEEEmembership{Fellow,~IEEE}
%		}
		\thanks{This work was supported in part by the National Natural Science Foundation for Distinguished Young Scholars (Grant No. 62325106), in part by the National Key Research and Development Program of China (Grant No. 2023YFB3811503), and in part by the National Natural Science Foundation of China (Grant No. 62031019).}
%		\thanks{All authors are with the Department of Electronic Engineering, Tsinghua University as well as Beijing National Research Center for Information Science and Technology (BNRist), Beijing 100084, China (E-mails: \{wzzc20, zja21\}@mails.tsinghua.edu.cn; \{yanyongli, daill\}@tsinghua.edu.cn).}
%		\thanks{Zhongzhichao~Wan, Jieao~Zhu, Yongli~Yan, and~Linglong~Dai are with the Department of Electronic Engineering, Tsinghua University, Beijing 100084, China (E-mails: \{wzzc20, zja21\}@mails.tsinghua.edu.cn; \{yanyongli, daill\}@tsinghua.edu.cn).} 
				\thanks{All authors are with the Department of Electronic Engineering, Tsinghua University, Beijing 100084, China (E-mails: \{wzzc20, zja21\}@mails.tsinghua.edu.cn; \{yanyongli, daill\}@tsinghua.edu.cn).} 
%		\thanks{Husheng Li is with the School of Aeronautics and Astronautics and
%	the Elmore Family School of Electrical and Computer Engineering, Purdue
%	University, West Lafayette, IN 47907, USA (Email: husheng@purdue.edu).}
%		\thanks{Lizhong Zheng is with the Department of Electrical Engineering and Computer Science, Massachusetts Institute of Technology, Cambridge, MA 02139, USA (Email: lizhong@mit.edu).}
			
	}
	\maketitle

\begin{abstract}
	Electromagnetic information theory (EIT) is one of the emerging topics for 6G communication due to its potential to reveal the performance limit of wireless communication systems. For EIT, one of the most important research directions is degree of freedom (DoF) analysis. Existing research works on DoF analysis for EIT focus on asymptotic conclusions of DoF, which do not well fit the practical wireless communication systems with finite spatial regions and finite frequency bandwidth. In this paper, we use the theoretical analyzing tools from Slepian concentration problem and extend them to three-dimensional space domain and four-dimensional space-time domain under electromagnetic constraints. Then we provide asymptotic DoF conclusions and non-asymptotic DoF analyzing scheme, which suits practical scenarios better, under different scenarios like three-dimensional antenna array. Moreover, we theoretically prove that the channel DoF is upper bounded by the proposed DoF of electromagnetic fields. Finally, we use numerical analysis to provide some insights about the optimal spatial sampling interval of the antenna array, the DoF of three-dimensional antenna array, the impact of unequal antenna spacing, the orthogonal space-time patterns, etc.
\end{abstract}


\section{Introduction}

Future sixth-generation (6G) wireless communication systems are expected to support various of emerging applications with considerable potential, e.g., autonomous vehicles, urban air mobility, and extended reality \cite{na2024operator}.
To meet the demand of high spectral density for 6G wireless communication systems, many promising technologies, including reconfigurable intelligent surfaces (RISs) \cite{basar2019wireless,wang2022location}, continuous-aperture multiple-input multiple-output (CAP-MIMO) \cite{huang2020holographic,zhang2023pattern}, and near-field communications \cite{cui2022near,wu2023multiple}, have been recently investigated. All of these technologies try to explore new sources of degrees of freedom (DoF) or capacity gain for the required performance improvement of 6G. The possible performance improvement actually comes from more accurate understanding and precise manipulation of electromagnetic fields which convey information \cite{chafii2023twelve}. Therefore, combining classical electromagnetic theory and information theory to provide modeling and capacity analysis tools is of great importance for exploring the fundamental performance limit of wireless communication systems, which leads to the interdisciplinary subject called electromagnetic information theory (EIT) \cite{migliore2018horse}. By integrating deterministic physical theory and stochastic mathematical theory \cite{zhu2022electromagnetic}, EIT is expected to provide new insights into system models, DoF, capacity limits, etc., from the electromagnetic perspective. 

The existing research directions of EIT includes channel modeling schemes \cite{gong2023holographic,wei2023tri,pizzo2023wide}, DoF analysis \cite{bucci1987spatial, bucci1989degrees,franceschetti2015landau}, mutual information and capacity analysis \cite{jensen2008capacity,jeon2017capacity,wan2022mutual}, etc.
Among these directions, DoF analysis is one of the most important parts, since it provides the upper bound of usable data streams and shows the required number of antennas, which guides the design of practical antenna arrays. Therefore, there is a fundamental question in EIT that how many DoFs can be provided by a wireless communication system constrained by electromagnetism? Moreover, how can we utilize these DoFs to transfer information? For the analysis of DoFs in EIT, one approach is considering bandwidth in the wavenumber domain.  The bandwidth above is called spatial bandwidth in \cite{bucci1987spatial}, which is similar to the widely-used bandwidth in time-frequency domain \cite{slepian1976bandwidth}. Identical to the time sampling scheme determined by classical bandwidth in the frequency domain, the bandwidth in wavenumber domain determines the sampling scheme in the space domain, which corresponds to the spatial DoF of the electromagnetic field. The spatial bandwidth of scattered fields was rigorously derived in \cite{bucci1987spatial}. It was shown that for a time-harmonic model with scatterers bounded in a sphere, the spatial bandwidth of the electromagnetic fields on an infinitely long observation line is linear with the frequency and the radius of the sphere. From the spatial bandwidth, the DoF obtained on a one-dimensional observation line is analyzed in \cite{bucci1989degrees} based on the conclusions of Slepian concentration problem and prolate spheroidal wave functions. For arbitrary received fields on two-dimensional surface, half-wavelength sampling scheme in the spatial domain was obtained in \cite{balanis2015antenna} when evanescent waves were discarded, which provided insightful results about the bound of DoF from direct spatial sampling. However, the spatial sampling requires infinite observation regions in the spatial domain, which is not consistent with practical wireless communication systems with finite spatial regions. 

For the DoF analysis considering finite spatial observation region, another approach follows Landau's eigenvalue theorem, which provides asymptotic DoF conclusions \cite{landau1980eigenvalue}. Here the asymptotic conclusion is achieved when the product of the Lebesgue measure of constraint regions in space and wavenumber domain tends to infinity, which implies either sufficiently large transceivers or sufficiently large frequency bandwidth. Following Landau's eigenvalue theorem, the DoF and spatial sampling scheme were discussed for two-dimensional antenna array \cite{pizzo2022nyquist}, which was based on the time-harmonic assumption and ignored the impact of the time domain. 
To further include the time domain in the wireless communication system, a model considering both space and time domain is used in \cite{franceschetti2015landau}, which shows the asymptotic DoF observed along a spatial cut-set boundary that separates transmitters and receivers. Although the existing works have already proposed many useful conclusions about the asymptotic DoF, there are still some important questions that remain to be answered. First of all, the non-asymptotic case which does not require sufficiently large transceivers or sufficiently large frequency bandwidth has not been discussed yet. Since practical wireless communication systems suit non-asymptotic cases better than asymptotic cases, discussing the DoF under non-asymptotic cases is of importance. Secondly, the DoF discussed here is the functional DoF of the transmitted/received electromagnetic fields, which is the dimension of the space constructed by all possible electromagnetic fields at one side of the transceiver\cite{zhu2022electromagnetic}. For practical wireless communication systems, we care more about the channel DoF which is determined by both transceivers and also the channel conditions. The relationship between the functional DoF and channel DoF remains to be discussed. 
Moreover, how to utilize the DoF to convey information needs to be answered, which corresponds to pattern design for the electromagnetic fields at the transceivers. 

To answer the above questions and provide a general DoF and pattern analyzing scheme for EIT, in this paper, we begin from the classical Slepian concentration problem which analyzes the model with limited time observation region and frequency bandwidth. We extend it to the three-dimensional space domain and four-dimensional space-time domain to obtain the DoF of the electromagnetic fields with given physical constraints\footnote{Simulation codes will be provided to reproduce the results in this paper: \url{http://oa.ee.tsinghua.edu.cn/dailinglong/publications/publications.html}.}.  Specifically, the contributions of this paper are summarized as follows:
\subsection{Contribution}
\begin{itemize}
	\item{We use the theoretical analyzing tools from Slepian concentration problem and extend it to three-dimensional space domain and four-dimensional space-time domain under electromagnetic constraints. Then we analyze the concentration problem involving time, frequency, space, and wavenumber domain for electromagnetic fields. From the analysis we can provide asymptotic DoF conclusions and non-asymptotic DoF analyzing schemes, which suit practical scenarios better, under different scenarios like three-dimensional antenna arrays. }
	\item{We provide rigorous definitions of the functional DoF analyzed above and channel DoF for continuous electromagnetic fields. Then, we provide theoretical proof about the relationship between the functional DoF and the channel DoF. It is shown that the channel DoF is upper bounded by the functional DoF at the transceivers.} 
	\item{We perform numerical simulations to show the non-asymptotic DoF under different scenarios, including three-dimensional antenna array and space-time cooperative information transmission model. The simulation results provide some insights about the optimal spatial sampling interval of antenna array, the DoF of three-dimensional antenna array, the impact of unequal antenna spacing, the orthogonal space-time patterns, etc. For example, our simulation results show that with narrow frequency bandwidth, half-wavelength antenna spacing is enough but not necessary to achieve the DoF upperbound.}	
\end{itemize}

\subsection{Notation}
Bold uppercase characters denote matrices;
bold lowercase characters denote vectors;
the dot $\cdot$ denotes the scalar product of two vectors, or the matrix-vector multiplication; 
$c$ is the speed of light in a vacuum; 
$\mathscr{F}[f(x)]$ denotes the Fourier transform of $f(x)$; 
$J_m(x)$ is the $m_{\rm th}$ order Bessel function of the first kind; 
The inner product of $f: \mathbb{R}^n \rightarrow \mathbb{C}$ and $f_i: \mathbb{R}^n \rightarrow \mathbb{C}$ on the domain $\mathcal{A}$ is denoted by $\langle f \ket{f_i}:= \int_{\mathcal{A}} f({\bf x}) f_i^*({\bf x}) {\rm d}^n{\bf x}$. The norm of $f$ is denoted by $\left\| f \right\|:=(\langle f \ket{f})^{1/2} $. 
	
	\section{DoF analysis based on Slepian concentration problem}
	\label{sec_theoretical_dof}
	In this part we will discuss the DoF of electromagnetic fields using the theoretical tools from Slepian concentration problem. 
	Slepian concentration problem considers the optimal representation of a signal bandlimited in the frequency domain and approximately bandlimited in time domain by using a set of orthogonal bases. The number of the base functions that need to be used to approximate the signal with a given error threshold represents the DoF of the signal space that satisfies the above conditions. We will first introduce the existing mathematical forms of Slepian concentration problem in one and two-dimensional cases and then extend them to the three-dimensional space domain and four-dimensional space-time domain under electromagnetic constraints. 
	
	\subsection{Slepian concentration problem in one and two-dimensional domain}
	\label{subsec_1d2d_Selpian}
	For typical one-dimensional time and frequency domain, the Slepian concentration problem is formulated in the following part.
	We have the one-dimensional Fourier transform and inverse Fourier transform
	\begin{equation}
		\begin{aligned}
			f(t) = (2\pi)^{-1}\int_{-\infty}^{+\infty}F(\omega)e^{{\rm j}\omega t}{\rm d}\omega,
		\end{aligned}
	\end{equation}
	and
	\begin{equation}
		\begin{aligned}
			F(\omega) = \int_{-\infty}^{+\infty}f(t)e^{-{\rm j}\omega t}{\rm d}t.
		\end{aligned}
	\end{equation}
	Then for a bandlimited signal $G(\omega) = \int_{-T}^{+T}g(t)e^{-{\rm j}\omega t}{\rm d}t$, the optimally concentrated signal is taken to be the one that maximizes $\lambda = \frac{\int_{-\Omega}^{\Omega}|G(\omega)|^2{\rm d}\omega}{\int_{-\infty}^{\infty}|G(\omega)|^2{\rm d}\omega}$, which leads to the following integral equation
	\begin{equation}
	\begin{aligned}
		\int_{-T}^{T}D(t,t')g(t'){\rm d}t' = \lambda g(t),
	\end{aligned}
	\end{equation}
	and
	\begin{equation}
		\begin{aligned}
			D(t,t') = (2\pi)^{-1}\int_{-\Omega}^\Omega e^{{\rm j}\omega(t-t')}{\rm d}\omega = \frac{\sin(\Omega(t-t'))}{\pi(t-t')}.
		\end{aligned}
	\end{equation}
	By analyzing the asymptotic behavior of the eigenvalue $\lambda$, it is known that the dimension of the signal space is asymptotically
	\begin{equation}
		\begin{aligned}
			N^{1D} = \sum_{i=1}^\infty \lambda_i = \int_{-T}^{T}D(t,t){\rm d}t = \frac{2\Omega T}{\pi},
		\end{aligned}
	\end{equation}
	when $\Omega T \rightarrow \infty$ \cite{slepian1961prolate}. We show in Fig. \ref{1d_omegaT} that when $\Omega T \rightarrow \infty$, the decay line of eigenvalues will have a cut-off transition. Nearly $N^{1D}$ eigenvalues are on the left side of the transition line, which tend to the same maximum value. Moreover, nearly all eigenvalues on the right side of the transition line tends to 0. Therefore, $N^{1D}$ is the asymptotic DoF of the signal space when $\Omega T \rightarrow \infty$.   
	
	\begin{figure}
		\centering 
		\includegraphics[width=0.7\textwidth]{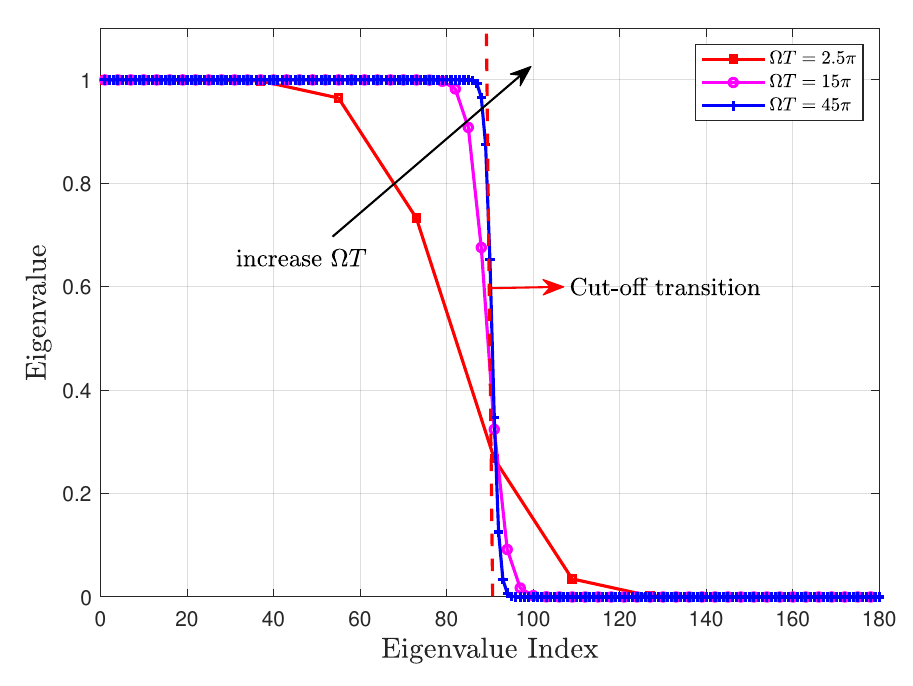} 
		\caption{The eigenvalues of one-dimensional Slepian concentration problem when $\Omega T$ increases.} 
		\label{1d_omegaT}
	\end{figure}
	
	For two-dimensional case, similar problem is formulated in the literature as follows:
	\begin{equation}
		\begin{aligned}
			f({\bf x})=(2\pi)^{-2}\int_{-\infty}^{+\infty}F({\bf k})e^{{\rm j}{\bf k}\cdot {\bf x}}{\rm d}^2{\bf k},
		\end{aligned}
		\label{eq_2d_fourier}
	\end{equation}
	and
	\begin{equation}
		\begin{aligned}
			F({\bf k}) = \int_{-\infty}^{+\infty}f({\bf x})e^{-{\rm j}{\bf k}\cdot {\bf x}}{\rm d}^2{\bf x}.
		\end{aligned}
	\end{equation}
	The integral equation corresponding to the two-dimensional eigenfunctions and eigenvalues is 
	\begin{equation}
		\begin{aligned}
			\int_{\mathcal{R}}D({\bf x},{\bf x}')f({\bf x}'){\rm d}^2{\bf x}' = \lambda f({\bf x}),
		\end{aligned}
	\end{equation}
	and
	\begin{equation}
		\begin{aligned}
			D({\bf x},{\bf x}') &= (2\pi)^{-2}\int_{\mathcal{K}} e^{{\rm j}{\bf k}\cdot({\bf x}-{\bf x}')}{\rm d}^3{\bf k} = \frac{K}{2\pi} \frac{J_1(K\left\| {\bf x}-{\bf x}' \right\|)}{\left\| {\bf x}-{\bf x}' \right\|},
		\end{aligned}
	\end{equation}
	where $\mathcal{K} = \{{\bf k} : \left\| {\bf k} \right\|\leqslant K\}$ is disk-shaped \cite{beylkin2007grids}. 
	The DoF when $K^2A \rightarrow \infty$ is 
	\begin{equation}
		\begin{aligned}
			N^{2D} = \sum_{i=1}^\infty \lambda_i = \int_{\mathcal{R}}D({\bf x},{\bf x}){\rm d}^2{\bf x} = \frac{K^2A}{4\pi},
		\end{aligned}
		\label{eq_2d_dof}
	\end{equation}
	where $A$ is the area of $\mathcal{R}$. 
	
	\subsection{Slepian concentration problem in three-dimensional domain}
	\label{subsec_3d_slepian}
	
	Now we will analyze the case with three-dimensional domain, which depicts the electromagnetic fields in the space. We consider the constraints of electromagnetic fields in space and wavenumber domain. In space domain the electromagnetic field is observed within the region of transceivers. In wavenumber domain the electromagnetic field is constrained by the frequency settings according to dispersion relation. There are two different cases of the constraints in wavenumber domain. The first one is with a frequency bandwidth, which corresponds to the constraint region as a ball in the wavenumber domain ($k_x^2+k_y^2+k_z^2 \leqslant k_0^2$) or a spherical shell in the wavenumber domain ($k_1^2 \leqslant k_x^2+k_y^2+k_z^2 \leqslant k_0^2$). The second one is on a single frequency point, which corresponds to the constraint region as a sphere in the wavenumber domain ($k_x^2+k_y^2+k_z^2=k_0^2$).  For the first case similar to what has been discussed in one-dimensional and two-dimensional examples, we have 
	\begin{equation}
		\begin{aligned}
			f({\bf x})=(2\pi)^{-3}\int_{-\infty}^{+\infty}F({\bf k})e^{{\rm j}{\bf k}\cdot {\bf x}}{\rm d}^3{\bf k},
		\end{aligned}
	\end{equation}
	and
	\begin{equation}
		\begin{aligned}
			F({\bf k}) = \int_{-\infty}^{+\infty}f({\bf x})e^{-{\rm j}{\bf k}\cdot {\bf x}}{\rm d}^3{\bf x}.
		\end{aligned}
	\end{equation}
	If the wavenumber domain constraint satisfies $k_x^2+k_y^2+k_z^2 \leqslant k_0^2$, we have 
	\begin{equation}
		\begin{aligned}
			D({\bf x},{\bf x}') &= (2\pi)^{-3}\int_{\mathcal{K}} e^{{\rm j}{\bf k}\cdot({\bf x}-{\bf x}')}{\rm d}^3{\bf k} \\&= (2\pi)^{-3} \int_{0}^{k_0} \int_{{\Omega}=\{\hat{\bf k}: \left\| \hat{\bf k} \right\|=1\}} e^{{\rm j}k\hat{\bf k}\cdot({\bf x}-{\bf x}')}k^2{\rm d}^2S{\rm d}k
			\\& = (2\pi)^{-3} 4\pi \int_0^{k_0} \frac{\sin k \left\| {\bf x}-{\bf x}' \right\|}{k \left\| {\bf x}-{\bf x}' \right\|}k^2 {\rm d}k
			\\& = \frac{1}{2\pi^2} \left( \frac{\sin k\left\| {\bf x}-{\bf x}' \right\|}{\left\| {\bf x}-{\bf x}' \right\|^3}-\frac{k\cos k\left\| {\bf x}-{\bf x}' \right\|}{\left\| {\bf x}-{\bf x}' \right\|^2} \right) \Bigg|_0^{k_0}
			\\& = \frac{1}{2\pi^2} \left( \frac{\sin k_0\left\| {\bf x}-{\bf x}' \right\|}{\left\| {\bf x}-{\bf x}' \right\|^3}-\frac{k_0\cos k_0\left\| {\bf x}-{\bf x}' \right\|}{\left\| {\bf x}-{\bf x}' \right\|^2} \right).
		\end{aligned}
		\label{equ_ball}
	\end{equation}
	The asymptotic DoF is 
	\begin{equation}
		\begin{aligned}
			N^{3D} = \sum_{i=1}^\infty \lambda_i = \int_{\mathcal{R}}D({\bf x},{\bf x}){\rm d}^3{\bf x} = \frac{V}{6\pi^2}k_0^3.
		\end{aligned}
	\end{equation}
	
	If the wavenumber domain constraint satisfies $k_1^2\leqslant k_x^2+k_y^2+k_z^2 \leqslant k_0^2$, similarly we have
	\begin{equation}
		\begin{aligned}
			D({\bf x},{\bf x}') =& \frac{1}{2\pi^2} \left( \frac{\sin k_0\left\| {\bf x}-{\bf x}' \right\|}{\left\| {\bf x}-{\bf x}' \right\|^3}-\frac{k_0\cos k_0\left\| {\bf x}-{\bf x}' \right\|}{\left\| {\bf x}-{\bf x}' \right\|^2} \right)
	-\frac{1}{2\pi^2} \left( \frac{\sin k_1\left\| {\bf x}-{\bf x}' \right\|}{\left\| {\bf x}-{\bf x}' \right\|^3}-\frac{k_1\cos k_1\left\| {\bf x}-{\bf x}' \right\|}{\left\| {\bf x}-{\bf x}' \right\|^2} \right).
		\end{aligned}
		\label{equ_spherical_shell}
	\end{equation}
	The asymptotic DoF is 
	\begin{equation}
	\begin{aligned}
		N^{3D} = \sum_{i=1}^\infty \lambda_i = \int_{\mathcal{R}}D({\bf x},{\bf x}){\rm d}^3{\bf x} = \frac{V}{6\pi^2}(k_0^3-k_1^3).
	\end{aligned}
	\end{equation}
	\begin{remark}
		Here the asymptotic DoF is achieved when $V(k_0^3-k_1^3)$ tends to infinity. For the single frequency point where $k_0=k_1$, the asymptotic DoF becomes 0, which seems no longer correct. We will do further discussion on it in the following part. 
		\end{remark}
	
%	For the single frequency point case, a delta function $\delta(k-k_0)$ need to be introduced to concentrate the function on the sphere in wavenumber domain. We denote $F({\bf k}) = F_1(\theta,\phi)\delta(k-k_0)$
%	
%		\begin{equation}
%		\begin{aligned}
%			f({\bf x})=(2\pi)^{-3}\int_{{\Omega}=\{\hat{\bf k}: \left\| \hat{\bf k} \right\|=1\}}k_0^2F_1(\theta,\phi)e^{{\rm j}k_0\hat{\bf k}\cdot {\bf x}}{\rm d}S,
%		\end{aligned}
%	\end{equation}
%	\begin{equation}
%		\begin{aligned}
%			F({\bf k}) = \int_{-\infty}^{+\infty}f({\bf x})e^{-{\rm j}{\bf k}\cdot {\bf x}}{\rm d}{\bf x},
%		\end{aligned}
%	\end{equation}
%	and
%	\begin{equation}
%		\begin{aligned}
%			F_1(\theta,\phi) &= 
%			\int_{0}^{\infty} \int_{-\infty}^{+\infty}f({\bf x})e^{-{\rm j}{\bf k}\cdot {\bf x}}{\rm d}{\bf x} {\rm d}k
%		\end{aligned}
%	\end{equation}
%	
%	Then similar to the previous cases, we have 
%	\begin{equation}
%		\begin{aligned}
%			D({\bf x},{\bf x}') &= (2\pi)^{-3}\int_{{\Omega}=\{\hat{\bf k}: \left\| \hat{\bf k} \right\|=1\}}k_0^2e^{{\rm j}k_0\hat{\bf k}\cdot {\bf x}}{\rm d}S
%			\\& = \frac{1}{2\pi^2} \frac{k_0\sin k_0 \left\| {\bf x}-{\bf x}' \right\|}{\left\|  {\bf x}-{\bf x}' \right\|}.
%			\end{aligned}
%		\end{equation}
%	The DoF accordingly is 
%		\begin{equation}
%		\begin{aligned}
%			N^{3D} = \sum_{i=1}^\infty \lambda_i = \int_{\mathcal{R}}D({\bf x},{\bf x}){\rm d}{\bf x} = \frac{V}{2\pi^2}k_0^2.
%		\end{aligned}
%	\end{equation}
	
	For the single frequency point case, we will use a function which converges to delta function in the $k$ axis to approach the sphere surface in wavenumber domain. We set $F({\bf k}) =F_1(\theta,\phi)F_2(k)$, where $\int_{0}^\infty F_2(k) {\rm d}k = 1$, then we have
	
	\begin{equation}
		\begin{aligned}
			f({\bf x})=&(2\pi)^{-3}\int_{{\Omega}=\{\hat{\bf k}: \left\| \hat{\bf k} \right\|=1\}}\int_{0}^{\infty}F_2(k)k^2F_1(\theta,\phi)e^{{\rm j}k\hat{\bf k}\cdot {\bf x}}{\rm d}k{\rm d}^2S,
		\end{aligned}
	\end{equation}
	and
	\begin{equation}
		\begin{aligned}
			F_1(\theta,\phi) &= 
			\frac{1}{F_2(k)} \int_{-\infty}^{+\infty}f({\bf x})e^{-{\rm j}{\bf k}\cdot {\bf x}}{\rm d}^3{\bf x}.
		\end{aligned}
	\end{equation}
	Then the kernel is represented by 
	\begin{equation}
		\begin{aligned}
			D({\bf x},{\bf x}') &= (2\pi)^{-3}\int_{{\Omega}=\{\hat{\bf k}: \left\| \hat{\bf k} \right\|=1\}}\int_{0}^{\infty}F_2(k)k^2e^{{\rm j}k\hat{\bf k}\cdot {\bf x}}{\rm d}k{\rm d}^2S.
		\end{aligned}
	\end{equation}
	The eigenvalues comes from the following integral equation
	\begin{equation}
		\lambda_i f_i({\bf x}) = \int_{\mathcal{R}}D({\bf x},{\bf x}')f_i({\bf x}'){\rm d}^3{\bf x}'.
		\label{eigenfunction_3d}
	\end{equation}
	Here we set $\{F_{2,n}(k)\}_{n=1}^{+\infty} = \frac{1}{2h}\left( 1-\frac{|k-k_0|}{2h} \right), |k-k_0|\leqslant 2h$, $h = \frac{1}{2^n}$ to be a set of functions that make $F({\bf k})$ focus on a sphere surface in the wavenumber domain. When $n \rightarrow \infty$, $F_2(k) \rightarrow \delta(k-k_0)$. To be clearer, we have
	 \begin{equation}
	 	\begin{aligned}
	 		&\int \frac{1}{2h}\left( 1-\frac{|k-k_0|}{2h} \right)k^2e^{{\rm j}k\hat{\bf k}\cdot {\bf x}}{\rm d}k
	 		= \left\{\begin{matrix}
	 		g_1(k)	& k \in [k_0,k_0+2h),\\
	 			g_2(k)	& k \in (k_0-2h,k_0),
	 		\end{matrix}\right.
	 	\end{aligned}
	 \end{equation}
	 where
	 \begin{equation}
	 	\begin{aligned}
	 		g_1(k) =& -\frac{1}{4 h^2 (\hat{\bf k}\cdot {\bf x})^4}e^{{\rm j} (\hat{\bf k}\cdot {\bf x}) k} \Bigg(2 {\rm j} h (\hat{\bf k}\cdot {\bf x}) \Big(k^2 (\hat{\bf k}\cdot {\bf x})^2+2 {\rm j} k (\hat{\bf k}\cdot {\bf x})-2\Big)+{\rm j} k^3 (\hat{\bf k}\cdot {\bf x})^2 (k_0-k)\\&+(\hat{\bf k}\cdot {\bf x})^2 k (3 k-2 k_0)+2 {\rm j} (\hat{\bf k}\cdot {\bf x}) (3 k-k_0)-6\Bigg),
		 \end{aligned}
     \end{equation}
     and
    \begin{equation}
    	\begin{aligned}
    		g_2(k) =& \frac{1}{4 h^2 (\hat{\bf k}\cdot {\bf x})^4}e^{{\rm j} (\hat{\bf k}\cdot {\bf x}) k} \Bigg(-2 {\rm j} h (\hat{\bf k}\cdot {\bf x}) \Big(k^2 (\hat{\bf k}\cdot {\bf x})^2+2 {\rm j} k (\hat{\bf k}\cdot {\bf x})-2\Big)+{\rm j} k^3 (\hat{\bf k}\cdot {\bf x})^2 (k_0-k)\\&+(\hat{\bf k}\cdot {\bf x})^2 k (3 k-2 k_0)+2 {\rm j} (\hat{\bf k}\cdot {\bf x}) (3 k-k_0)-6\Bigg).
    	\end{aligned}
    \end{equation} 
    Therefore, we know that 
    \begin{equation}
    	\begin{aligned}
    		\int_{k_0-2h}^{k_0+2h} \frac{1}{2h}\left( 1-\frac{|k-k_0|}{2h} \right)k^2e^{{\rm j}k\hat{\bf k}\cdot {\bf x}}{\rm d}k
    		=& -\frac{1}{4 h^2 (\hat{\bf k}\cdot {\bf x})^4}e^{{\rm j} (\hat{\bf k}\cdot {\bf x}) k_0-2 {\rm j} h (\hat{\bf k}\cdot {\bf x})} 
    		\Bigg(e^{4 {\rm j} h (\hat{\bf k}\cdot {\bf x})} \Big(4 h^2 (\hat{\bf k}\cdot {\bf x})^2+\\&4 h (\hat{\bf k}\cdot {\bf x}) ((\hat{\bf k}\cdot {\bf x}) k_0+2 {\rm j})+(\hat{\bf k}\cdot {\bf x})^2 k_0^2+4 {\rm j} (\hat{\bf k}\cdot {\bf x}) k_0-6\Big)
    		\\&+4 h^2 (\hat{\bf k}\cdot {\bf x})^2-2 e^{2 {\rm j} h (\hat{\bf k}\cdot {\bf x})} \Big((\hat{\bf k}\cdot {\bf x})^2 k_0^2+4 {\rm j} (\hat{\bf k}\cdot {\bf x}) k_0-6\Big)\\&-4 h (\hat{\bf k}\cdot {\bf x}) ((\hat{\bf k}\cdot {\bf x}) k_0+2 {\rm j})+(\hat{\bf k}\cdot {\bf x})^2 k_0^2+4 {\rm j} (\hat{\bf k}\cdot {\bf x}) k_0-6\Bigg).
    	\end{aligned}
    \end{equation}
    When we let $h\rightarrow 0$, $\int_{k_0-2h}^{k_0+2h} \frac{1}{2h}\left( 1-\frac{|k-k_0|}{2h} \right)k^2e^{{\rm j}k\hat{\bf k}\cdot {\bf x}}{\rm d}k \rightarrow k_0^2 e^{{\rm j}k_0 \hat{\bf k}\cdot {\bf x}}$, which makes the kernel function as 
    \begin{equation}
    	\begin{aligned}
    		D({\bf x},{\bf x}') =&  \frac{1}{2\pi^2} \frac{k_0\sin k_0 \left\| {\bf x}-{\bf x}' \right\|}{\left\|  {\bf x}-{\bf x}' \right\|}.
    	\end{aligned}
    \end{equation}
    Then the integral equation transforms to
    	\begin{equation}
    	\lambda_i f_i({\bf x}) = \int_{\mathcal{R}} \frac{1}{2\pi^2} \frac{k_0\sin k_0 \left\| {\bf x}-{\bf x}' \right\|}{\left\|  {\bf x}-{\bf x}' \right\|}f_i({\bf x}'){\rm d}^3{\bf x}'.
    	\label{equ_sphere_surface}
    \end{equation}
    By solving this integral equation we obtain the eigenvalues that determine the number of base functions required to construct the received electric field, which can be viewed as the dimension of the subspace constructed by the possible received electric field. To be more specific, the eigenvalues determine the "importance" of the eigenfunctions as bases of the electromagnetic fields which satisfy the concentration conditions. From the Slepian concentration problem we have a set of orthogonal eigenfunctions $f_i$ and eigenvalues $\lambda_i$ for functions band-limited in wavenumber domain and approximately band-limited in the space domain. If we use these functions to approximate arbitrary function $f$ band-limited in wavenumber domain, we have 
		\begin{equation}
	\begin{aligned}
		\left\| f-\sum_{i=1}^n a_i f_i  \right\|^2 &= 
		\langle f-\sum_{i=1}^n a_i f_i \ket{f-\sum_{i=1}^n a_i f_i}
		\\& = \langle f \ket{f} - \sum_{i=1}^n a_i^* \langle f \ket{f_i} -\sum_{i=1}^n a_i \langle f_i \ket{f} + \sum_{i=1}^n |a_i|^2 \langle f_i \ket{f_i}
		\\& = \left\| f \right\|^2 + \sum_{i=1}^n \left| a_i \left\| f_i \right\| - \frac{ \langle f \ket{f_i}}{\left\| f_i \right\|} \right|^2 - \sum_{i=1}^n \frac{ \langle f \ket{f_i}^2}{\left\| f_i \right\|^2}
		\\&\geqslant  \left\| f \right\|^2 - \sum_{i=1}^n \frac{ \langle f \ket{f_i}^2}{\left\| f_i \right\|^2}
		\\&=\sum_{i=n+1}^{\infty} b_i^2 \lambda_i ,
	\end{aligned}
	\label{eqn_min_f_approx}
	\end{equation}
	where the inner product is defined as $\langle f \ket{f_i}:= \int_{\mathcal{R}} f({\bf x}) f_i^*({\bf x}) {\rm d}^3{\bf x}$, $\mathcal{R}$ is the concentration region in the space domain, $b_i = \frac{ \langle f \ket{f_i}}{\left\| f_i \right\|^2}= \frac{1}{\lambda_i} \langle f \ket{f_i}$, the equality is achieved when $a_i=b_i$. Then it is obvious that by using the eigenfunctions with large eigenvalues as bases, we can approximate arbitrary concentrated function with tolerable error. Although the space constructed by all concentrated functions has infinite dimensions, we only care about part of them. The eigenfunctions correspond to the dimensions we care about can be used as waveform patterns to construct any required electromagnetic field in wireless communication process.       
   
	In the above part we are analyzing the scenario where the direction of ${\bf k}$ is not restricted, which means that the incident wave to the receiver may have arbitrary direction, corresponding to general channel settings. Given a specific channel, the direction of the incident wave to the receiver may be constrained, which adds further restrictions on ${\bf k}$. Specifically, new constraints ${\Omega}_1 \subset {\Omega}=\{\hat{\bf k}: \left\| \hat{\bf k} \right\|=1\}$ on $\hat{\bf k}$ can be introduced, which is determined by the concentration region of the electromagnetic field caused by the scattering channel. Then we can have the following transform
	\begin{equation}
		\begin{aligned}
			f({\bf x})=(2\pi)^{-3}\int_0^{k_0} \int_{{\Omega}_1 \subset {\Omega}=\{\hat{\bf k}: \left\| \hat{\bf k} \right\|=1\}} F({\bf k})e^{{\rm j}{\bf k}\cdot {\bf x}}{\rm d}^2S {\rm d}{k},
		\end{aligned}
		\label{eqn_wavenumber_constraint1}
	\end{equation}
	or
	\begin{equation}
		\begin{aligned}
			f({\bf x})=(2\pi)^{-3}\int_{k_1}^{k_0} \int_{{\Omega}_1 \subset {\Omega}=\{\hat{\bf k}: \left\| \hat{\bf k} \right\|=1\}} F({\bf k})e^{{\rm j}{\bf k}\cdot {\bf x}}{\rm d}^2S {\rm d}{k}.
		\end{aligned}
		\label{eqn_wavenumber_constraint2}
	\end{equation}
	In single-frequency scenario, it corresponds to
	\begin{equation}
		\begin{aligned}
			f({\bf x})=(2\pi)^{-3}\int_{{\Omega}_1 \subset {\Omega}=\{\hat{\bf k}: \left\| \hat{\bf k} \right\|=1\}}k_0^2F_1(\theta,\phi)e^{{\rm j}k_0\hat{\bf k}\cdot {\bf x}}{\rm d}^2S.
		\end{aligned}
		\label{eqn_wavenumber_constraint3}
	\end{equation}

Since ${\Omega}_1$ may have complicated shapes, which makes the aforementioned scheme that integrating from wavenumber domain hard to analyze, we need to find another way to solve the concentration problem. It is shown that the solving the integral equation of the concentration problem from either of the two domains that are Fourier and inverse Fourier transforms of each other yields consistent eigenvalues\cite{beylkin2007grids}. Therefore, if the concentration region in the space domain has regular shapes, we can obtain the kernel function by integrating from the spatial domain, which is easier. Assume that the spatial concentration region is a ball which satisfies $\left\| {\bf x} \right\|\leqslant r_0$, we can obtain the kernel function as
\begin{equation}
	\begin{aligned}
		D(\hat{\bf k},\hat{\bf k}') &= (2\pi)^{-3}\int_{\mathcal{K}} e^{{\rm j}k{\bf x}\cdot(\hat{\bf k}-\hat{\bf k}')}{\rm d}^3{\bf x}\Bigg|_{k=k_0} 
		\\&= \frac{1}{2\pi^2} \left( \frac{\sin r_0k_0\left\| \hat{\bf k}-\hat{\bf k}' \right\|}{k_0^3\left\| \hat{\bf k}-\hat{\bf k}' \right\|^3}-\frac{r_0\cos r_0k_0\left\| \hat{\bf k}-\hat{\bf k}' \right\|}{k_0^2\left\| \hat{\bf k}-\hat{\bf k}' \right\|^2} \right).
	\end{aligned}
\end{equation}
If the spatial concentration region is a cuboid, which satisfies $\left| x_x \right|\leqslant a_x$, $\left| x_y \right|\leqslant a_y$ and $\left| x_z \right|\leqslant a_z$, we can obtain the kernel function as
\begin{equation}
	\begin{aligned}
		D(\hat{\bf k},\hat{\bf k}') &= (2\pi)^{-3}\int_{-a_x}^{a_x}\int_{-a_y}^{a_y}\int_{-a_z}^{a_z} e^{{\rm j}k{\bf x}\cdot(\hat{\bf k}-\hat{\bf k}')}{\rm d}^3{\bf x}\Bigg|_{k=k_0}
		\\& = \frac{\sin(a_x(k_x-k_x'))}{\pi(k_x-k_x')}\frac{\sin(a_y(k_y-k_y'))}{\pi(k_y-k_y')}\frac{\sin(a_z(k_z-k_z'))}{\pi(k_z-k_z')}\Bigg|_{k=k_0}.
	\end{aligned}
\end{equation}
Then, corresponding to \eqref{eqn_wavenumber_constraint1}, the concentrated functions satisfy the following integral equation
\begin{equation}
	\lambda_i \phi_i({\bf k}) = \int_0^{k_0}\int_{{\Omega}_1 \subset {\Omega}=\{\hat{\bf k}': \left\| \hat{\bf k}' \right\|=1\}} D({\bf k},{\bf k}')\phi_i({\bf k}'){\rm d}^3{\bf k}'.
	\label{eqn_ball_wavenumber_constraint}
\end{equation}
Corresponding to \eqref{eqn_wavenumber_constraint2}, the concentrated functions satisfy the following integral equation
\begin{equation}
	\lambda_i \phi_i({\bf k}) = \int_{k_1}^{k_0}\int_{{\Omega}_1 \subset {\Omega}=\{\hat{\bf k}': \left\| \hat{\bf k}' \right\|=1\}} D({\bf k},{\bf k}')\phi_i({\bf k}'){\rm d}^3{\bf k}'.
\end{equation}
Corresponding to \eqref{eqn_wavenumber_constraint3}, the concentrated functions satisfy the following integral equation
\begin{equation}
	\lambda_i \phi_i(\hat{\bf k}) = \int_{{\Omega}_1 \subset {\Omega}=\{\hat{\bf k}': \left\| \hat{\bf k}' \right\|=1\}} D(\hat{\bf k},\hat{\bf k}')\phi_i(\hat{\bf k}'){\rm d}^2\hat{\bf k}'.
\end{equation}
%The DoF is then 
%\begin{equation}
%	\begin{aligned}
%		N^{3D} = \sum_{i=1}^\infty \lambda_i = \int_{\Omega_1}D(\hat{\bf k},\hat{\bf k}){\rm d}\hat{\bf k} = \frac{\mathcal{A}_{\Omega_1}a_xa_ya_z}{\pi^3}.
%	\end{aligned}
%\end{equation}
	
%	For two-dimensional antenna surface, if we only consider $x$ and $y$ axes and uses the two-dimensional Fourier transform, the Slepian's concentration problem follows (\ref{eq_2d_fourier}) to (\ref{eq_2d_dof}), which coincides the result in \cite{pizzo2022nyquist} using Landau's formula. However, this result is not accurate because an extra term should be added from the $z$ axis omitted in the two-dimensional Fourier transform. From the three-dimensional case, now we try to modify the derivation process. Since we have $\frac{\partial k_z}{\partial k_x} = \frac{-k_x}{\sqrt{k_0^2-k_x^2-k_y^2}}$ and $\frac{\partial k_z}{\partial k_y} = \frac{-k_y}{\sqrt{k_0^2-k_x^2-k_y^2}}$, we can obtain
%		\begin{equation}
%	\begin{aligned}
%		f({\bf x})&=(2\pi)^{-3}\int_{{\Omega}=\{\hat{\bf k}: \left\| \hat{\bf k} \right\|=k_0^2\}}F({k}_x,{k}_y,{k}_z)e^{{\rm j}{\bf k}\cdot {\bf x}}{\rm d}S
%		\\& = (2\pi)^{-3}\int_{k_x^2+k_y^2\leqslant k_0^2}\frac{k_0e^{{\rm j}\sqrt{k_0^2-k_x^2-k_y^2}x_z}F_2(k_x,k_y)}{\sqrt{k_0^2-k_x^2-k_y^2}}\\&~~~~e^{{\rm j}(k_xx_x+k_yx_y)}{\rm d}k_x{\rm d}k_y,
%	\end{aligned}
%\end{equation}
%and
%	\begin{equation}
%	\begin{aligned}
%		F({\bf k}) = \int_{-\infty}^{+\infty}f({\bf x})e^{-{\rm j}{\bf k}\cdot {\bf x}}{\rm d}{\bf x},
%	\end{aligned}
%\end{equation}
%	
	
	\subsection{Slepian's concentration problem in four-dimensional space-time domain}
	
	After discussing the theoretical analysis in the three-dimensional domain, we now go further into the four-dimensional space-time domain, which fully characterizes the electromagnetic fields. Moreover, only with the change in the time domain can we transmit information in practical wireless communication systems.
	
	For ease of analysis, we introduce the following theorem from \cite{ihara1993information}.
	
	\begin{theorem}
		\cite{ihara1993information} Let $\mathcal{H}$ be a Hilbert space. Let $\mathcal{H}_1$ and $\mathcal{H}_2$ be two subspaces of $\mathcal{H}$. We have $\Pi_1$ and $\Pi_2$ be two projection operators that satisfy $\Pi_1 x \in \mathcal{H}_1$ and $\Pi_2 x \in \mathcal{H}_2$ for $x \in \mathcal{H}$. Let $A = \Pi_1 \Pi_2 \Pi_1$ and $B = \Pi_2 \Pi_1 \Pi_2$. Then, there exist complete orthonormal bases $\{\phi_n \}_{n=1}^{\infty}$ for $\mathcal{H}_1$ and $\{\psi_n \}_{n=1}^{\infty}$ for $\mathcal{H}_2$ that span the eigenspaces of $A$ and $B$ separately. The eigenvalues of $A$ and $B$ form the same sequence $\{\lambda_n \}_{n=1}^{+\infty}$, which satisfies
		\begin{subequations}
			\begin{align} 
				&A\phi_n = \lambda_n \phi_n, \label{Aphi}\\
				& B \psi_n = \lambda_n \psi_n,  \label{Bpsi}\\
				& \langle \phi_m \ket{\psi_n} = \sqrt{\lambda_n} \delta_{mn},
			\end{align}
		\end{subequations}
		for arbitrary $m,n\in \mathbb{Z}^+$.
	\end{theorem}
	\begin{remark}
		From the projection operators we can obtain the same mathematical forms as (\ref{eigenfunction_3d}), leading to the same eigenfunctions and eigenvalues. 
	\end{remark}
	Now we use {\bf Theorem 1} to analyze the DoF of the received electromagnetic fields in the four-dimensional space-time domain. In our considered scenario, the Hilbert space $\mathcal{H}_1$ contains the electromagnetic fields that are constrained in space domain $\mathcal{R}$ and constrained in time domain $\mathcal{T}$. The Hilbert space $\mathcal{H}_2$ contains the electromagnetic fields that are constrained in frequency domain $\mathcal{W}$ and constrained in wavenumber domain $\mathcal{K}$. Moreover, $\mathcal{H}_2$ also satisfies $w^2 = c^2\left\| {\bf k} \right\|^2 $ according to electromagnetic constraints. Therefore, we have $\Pi_1f = \mathcal{B}_1 f $ and $\Pi_2 f = \mathcal{F}^{-1} \mathcal{B}_2 \mathcal{F} f$, where 
	\begin{equation}
		\begin{aligned}
			\mathcal{B}_1 f = \left\{\begin{matrix}
			f	&  t\in \mathcal{T} \& {\bf x} \in \mathcal{R},  \\
			0	&  {\rm otherwise},
			\end{matrix}\right.
		\end{aligned}
	\end{equation}
	and
		\begin{equation}
		\begin{aligned}
			\mathcal{B}_2 f = \left\{\begin{matrix}
				f	&  w= \pm c \left\| {\bf k} \right\| \& {\bf k} \in \mathcal{K},  \\
				0	&  {\rm otherwise}.
			\end{matrix}\right.
		\end{aligned}
	\end{equation}
	The concentration regions $\mathcal{H}_1$ and $\mathcal{H}_2$ constructed from the electromagnetic constraints are shown in Fig. \ref{H1H2projection}.
	
	Noted that we can abbreviate $B$ as $\Pi_2 \Pi_1$ since the eigenfunctions $\psi_n$ are in $\mathcal{H}_2$. We have 
	\begin{equation}
		\begin{aligned}
		(\Pi_2 \Pi_1 f)({\bf x}',t') &=  (\mathcal{F}^{-1} \mathcal{B}_2 \mathcal{F} \mathcal{B}_1f)({\bf x},t)
		\\& = \int_{w= \pm c \left\| {\bf k} \right\| \& {\bf k} \in \mathcal{K}} e^{{\rm j}{\bf k}\cdot{\bf x}-{\rm j}wt} \int_{t\in \mathcal{T} \& {\bf x} \in \mathcal{R}} e^{-{\rm j}{\bf k}\cdot {\bf x}'+{\rm j}wt' } f({\bf x}',t'){\rm d}^3{\bf x}'{\rm d}t' {\rm d}^3{\bf k}{\rm d}w
		\\& = \int_{t\in \mathcal{T} \& {\bf x} \in \mathcal{R}}  \int_{w= \pm c \left\| {\bf k} \right\| \& {\bf k} \in \mathcal{K}} e^{{\rm j}{\bf k}\cdot({\bf x}-{\bf x}')-{\rm j}w(t-t')}
		 {\rm d}^3{\bf k}{\rm d}w f({\bf x}',t'){\rm d}^3{\bf x}'{\rm d}t',
		\end{aligned}
	\end{equation}
	and the integral equation is $(\Pi_2 \Pi_1 f)({\bf x}',t') = \lambda f({\bf x},t)$.

	\begin{figure}
		\centering 
		\includegraphics[width=0.7\textwidth]{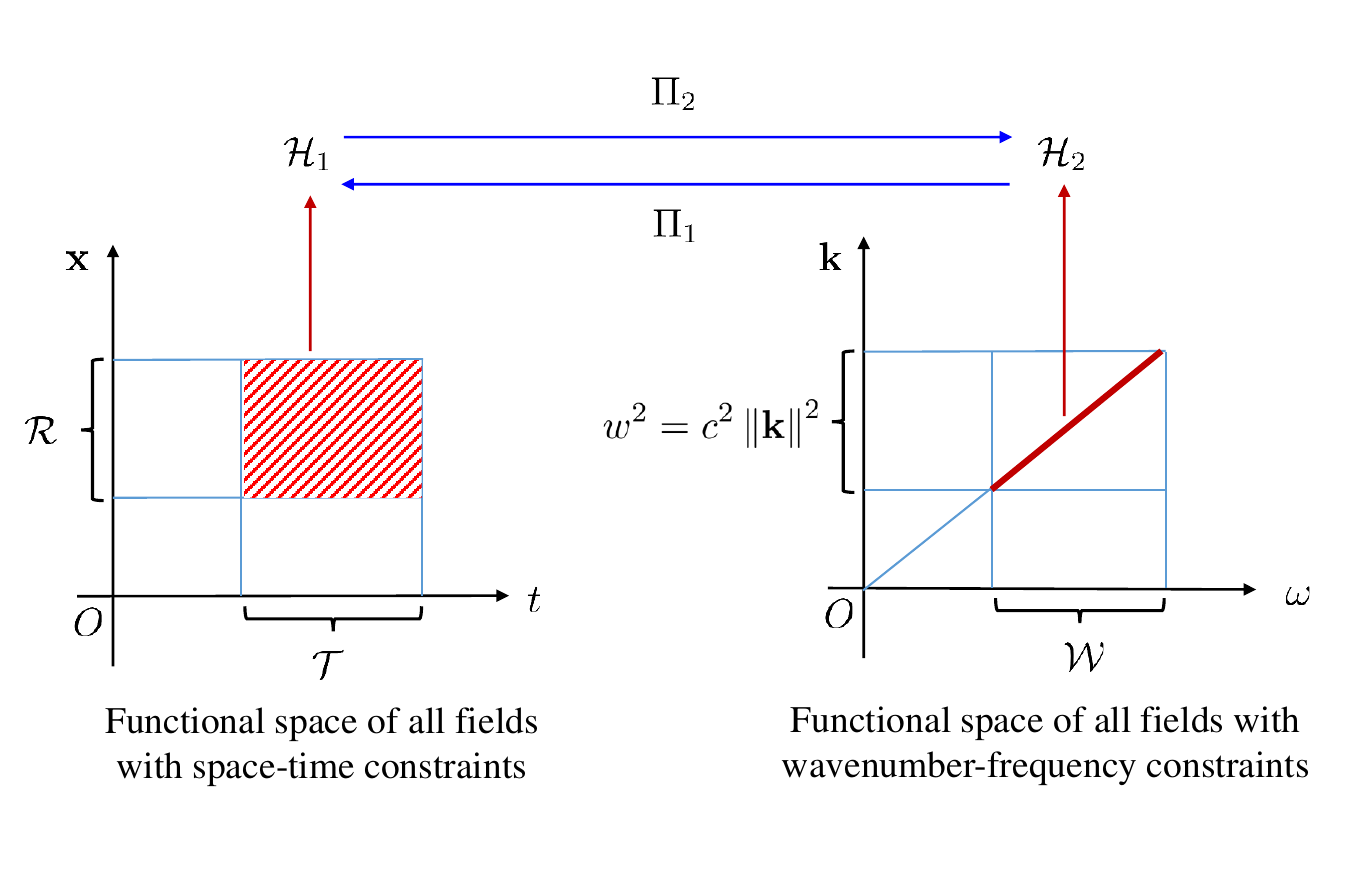} 
		\caption{The two functional spaces $\mathcal{H}_1$ and $\mathcal{H}_2$ include all fields with space-time constraints and all fields with wavenumber-frequency constraints separately. $\mathcal{H}_1$ and $\mathcal{H}_2$ are subspaces of $\mathcal{H}$. Through projection operators $\Pi_1$ and $\Pi_2$ we can project arbitrary $x \in \mathcal{H}$ onto $\mathcal{H}_1$ and $\mathcal{H}_2$.  } 
		\label{H1H2projection}
	\end{figure}

	If we only consider a single time point, we set $t=t'=t_0$, the eigenfunction can be simplified as 
	\begin{equation}
		\begin{aligned}
			\int_{{\bf x} \in \mathcal{R}}  \int_{ {\bf k} \in \mathcal{K}} e^{{\rm j}{\bf k}\cdot({\bf x}-{\bf x}')}
			{\rm d}^3{\bf k} f({\bf x}',t_0){\rm d}^3{\bf x}' = \lambda f({\bf x},t_0).
		\end{aligned}
	\end{equation}
	It is worth noting that by omitting time interval, the eigenfunction degenerates to the form in {\bf Section \ref{subsec_3d_slepian}}, which verifies the above analysis in three-dimensional space. 
	
	Now we consider using a continuous time interval. Moreover, we assume that the frequency band is $\omega \in [-ck_0, ck_0]$ and the wavenumber domain satisfies $\left\| {\bf k} \right\| \leqslant k_0$. Then we have 
	\begin{equation}
		\begin{aligned}
			\int_{w= \pm c \left\| {\bf k} \right\| \& {\bf k} \in \mathcal{K}} e^{{\rm j}{\bf k}\cdot({\bf x}-{\bf x}')-{\rm j}w(t-t')}
			{\rm d}^3{\bf k}{\rm d}w
			= & \int_{0}^{k_0} \int_{{\Omega}=\{\hat{\bf k}: \left\| \hat{\bf k} \right\|=1\}} e^{{\rm j}k \hat{\bf k}\cdot({\bf x}-{\bf x}')} {\rm d}^2S 2\cos(ck(t-t')) k^2{\rm d}k
			\\ = & \int_{0}^{k_0} 4\pi \frac{\sin k \left\| {\bf x}-{\bf x}' \right\|}{k \left\| {\bf x}-{\bf x}' \right\|}  2\cos(ck(t-t'))k^2{\rm d}k,
		\end{aligned}
	\end{equation}
	which can be further simplified to (\ref{equ_kernel_expression}).
	By substituting the kernel function in $\Pi_2 \Pi_1 f({\bf x}',t') = \lambda f({\bf x},t)$, we can obtain the eigenvalues and eigenfunctions which correspond to DoF and orthogonal patterns for the electromagnetic fields.
	
			\begin{figure*}
		\begin{equation}
			\begin{aligned}
				%				&\frac{2\pi e^{-{\rm j}ck_0(t-t')}}{\left\| {\bf x}-{\bf x}' \right\|(\left\| {\bf x}-{\bf x}' \right\|^2-c^2(t-t')^2)^2} \Bigg(2{\rm j}c \left\| {\bf x}-{\bf x}' \right\|(t-t') e^{{\rm j}ck_0(t-t')}
				%				+ \left\| {\bf x}-{\bf x}' \right\| \Big( -k_0\left\| {\bf x}-{\bf x}' \right\|^2 + c(t-t')(-2{\rm j}+c(t-t')k_0)  \Big)
				%				\\&\cos(k_0 \left\| {\bf x}-{\bf x}' \right\|) + \Big( \left\| {\bf x}-{\bf x}' \right\|^2 (1-{\rm j}ck_0(t-t')) + c^2(t-t')^2(1+{\rm j}ck_0(t-t')) \Big) \sin(k_0 \left\| {\bf x}-{\bf x}' \right\|)\Bigg)
				-\frac{4\pi\Big(- \left\| {\bf x}-{\bf x}' \right\| + \left\| {\bf x}-{\bf x}' \right\| \cos(k_0\left\| {\bf x}-{\bf x}' \right\|)\cos(ck_0(t-t')) + c(t-t')\sin(k_0\left\| {\bf x}-{\bf x}' \right\|)\sin(ck_0(t-t'))  \Big)}{\left\| {\bf x}-{\bf x}' \right\|(\left\| {\bf x}-{\bf x}' \right\|^2-c^2(t-t')^2)^2}
			\end{aligned}
			\label{equ_kernel_expression}
		\end{equation}
		{\noindent} \rule[-10pt]{18cm}{0.05em}
	\end{figure*}

	\section{Relationship between functional DoF and channel DoF}
	
	In this section we will discuss the relationship between functional DoF of electromagnetic fields and the channel DoF, which shows that how the conclusions about functional DoF obtained from Slepian concentration problem or Landau's eigenvalue theorem can affect the performance of wireless communication systems. First we will give rigorous mathematical definitions of functional DoF and channel DoF. Then we will provide a theorem showing the relationship between them, which reveals that the channel DoF is upper bounded by the functional DoF at the transceivers.
	
	\begin{definition}
		(functional DoF) Let $\mathcal{H}$ be a Hilbert space and $\{f_i\}_{i=1}^n \in \mathcal{H}$ be a set of orthogonal base functions. If for arbitrary function $f \in \mathcal{H}$ that satisfies $\left\| f \right\|=1$, there always exists a set of numbers $a_i$ that 
		$\left\| f-\sum_{i=1}^n a_i f_i  \right\| \leqslant \epsilon$. The minimum $n$ that satisfies the above conditions is the functional DoF of $\mathcal{H}$ with error $\epsilon$, i.e.,
		\begin{equation}
			\mathsf{fDoF}_{\epsilon} (\mathcal{H}) = {\rm min}\{ n: \underset{\mathcal{H}_n \subset \mathcal{H}}{\rm inf} \sup_{f: \left\| f \right\|=1} \inf_{g \in \mathcal{H}_n} \left\| f-g \right\| \leqslant \epsilon  \},
		\end{equation}
		where $\mathcal{H}_n$ is an $n$-dimensional subspace of $\mathcal{H}$.
	\end{definition}
	
	\begin{remark}
		(relationship between the functional DoF and the Slepian concentration problem) From the Slepian concentration problem we have a set of orthogonal eigen-functions $f_i \in \mathcal{H}$ and eigenvalues $\lambda_i$ for functions band-limited in wavenumber domain and approximately band-limited in the space domain. From \eqref{eqn_min_f_approx} we know that these functions satisfy the condition that $\underset{a_i^n}{\rm min}\left\| f- \sum_{i=1}^n a_i f_i    \right\|^2 = \sum_{i=n+1}^{+\infty}b_i^2 \lambda_i $, where $b_i =\frac{ \langle f \ket{f_i}}{\left\| f_i \right\|^2}= \frac{1}{\lambda_i} \langle f \ket{f_i}$ and the minimum is achieved when $a_i=b_i$. Moreover, we have
		\begin{equation}
			\begin{aligned}
				\underset{a_i^n}{\rm min}\left\| f- \sum_{i=1}^n a_i f_i    \right\|^2	\leqslant (\sum_{i=n+1}^{+\infty}b_i^2) \lambda_{n+1} \leqslant \lambda_{n+1}.
			\end{aligned}
		\end{equation}  
		Therefore, the number of large eigenvalues in the Slepian concentration problem determines the functional DoF of the space $\mathcal{H}$, i.e.,
		\begin{equation}
			\begin{aligned}
				\mathsf{fDoF}_{\epsilon} (\mathcal{H}) \leqslant \#\{n: \lambda_{n+1} \geqslant \epsilon^2\}.
			\end{aligned}
		\end{equation} 
	\end{remark}
	
	\begin{definition}
		(channel DoF) Let $\mathcal{H}_{\rm t}$ and $\mathcal{H}_{\rm r}$ be two Hilbert spaces and $P:\mathcal{H}_{\rm t} \rightarrow \mathcal{H}_{\rm r}$ be an operator. Let $\{g_i\}_{i=1}^{+\infty}$ and $\{f_i\}_{i=1}^{+\infty}$ be a set of orthogonal bases for $\mathcal{H}_{\rm t}$ and $\mathcal{H}_{\rm r}$ separately. Let $\{P_i\}_{i=1}^n: \mathcal{H}_{\rm t} \rightarrow \mathcal{H}_{\rm r}$ be a set of rank-1 operators. 
		If for arbitrary $g \in \mathcal{H}_{\rm t}$ that satisfies $\left\| g \right\|=1$, there always exists a set of numbers $b_i$ that satisfies $\left\| Pg-\sum_{i=1}^{n} b_iP_ig  \right\| \leqslant \epsilon$. The minimum $n$ that satisfies the above conditions is the channel DoF of the operator $P: \mathcal{H}_{\rm t} \rightarrow \mathcal{H}_{\rm r}$ with error $\epsilon$, i.e., 
		\begin{equation}
			\begin{aligned}
			\mathsf{cDoF}_{\epsilon} (\mathcal{H}_{\rm t},\mathcal{H}_{\rm r},P) =& {\rm min}\{ n: \underset{\{P_i\}_{i=1}^n: \mathcal{H}_{\rm t} \rightarrow \mathcal{H}_{\rm r}}{\rm inf} \sup_{g: \left\| g \right\|=1} \inf_{\{b_i\}_{i=1}^n} \left\| Pg-\sum_{i=1}^{n} b_iP_ig  \right\| \leqslant \epsilon  \}.
			\end{aligned}
		\end{equation}
	\end{definition}
	
	\begin{remark}
		(Motivation of the definition of the channel DoF based on operator)
		The definition of the channel DoF here is a natural extension of the classical channel DoF based on matrices. It is well known that for a channel matrix $\bf{P}$ the channel DoF can be obtained by analyzing the singular value decomposition (SVD) of $\bf{P}$, leading to ${\bf P} = {\bf U}\boldsymbol{\Lambda}\bf{V} = \sum_{i=1}^n \lambda_i {\bf u}_i{\bf v}_i^{\rm T}$. 
		If we can approximating ${\bf P}$ by the linear combination of a number of matrices ${\bf u}_i{\bf v}_i^{\rm T}$, the DoF of the channel is found. For general electromagnetic environment we can assume that the integral operator $P$ that represents the channel is "smooth" to ensure that the kernel function and the first-order derivative of the kernel function are square-integrable on the source region. Then we know that the operator $P$ is a trace-class operator\cite{bornemann2010numerical}. Then we can decompose $P$ as the sum of a set of operators multiplying the eigenvalues of $P$, similar to the SVD of matrices. Then in {\bf Definition} 2 the decomposition of matrices is extended to the decomposition of operators.
	\end{remark}
	
				\begin{figure*}
		\centering 
		\includegraphics[width=1\textwidth]{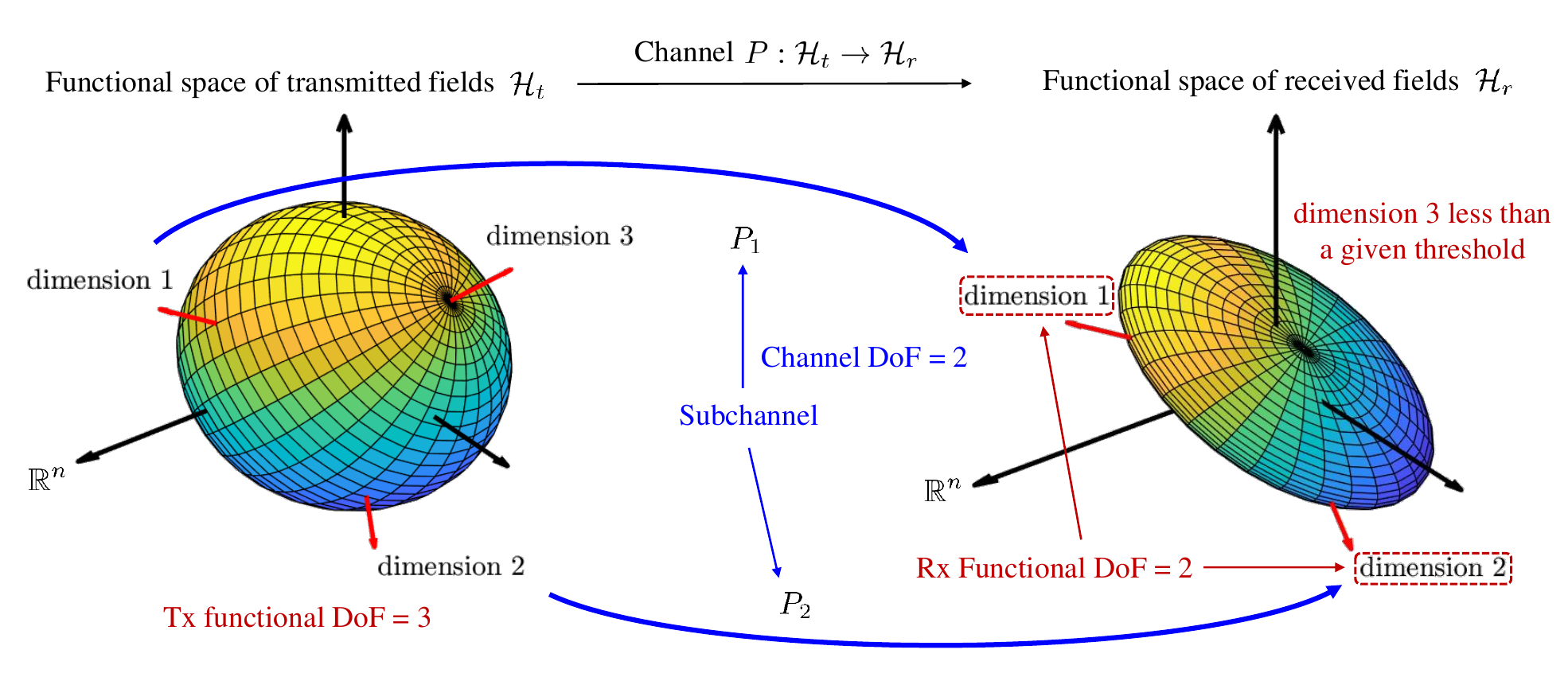} 
		\caption{Description of the definition of functional DoF and channel DoF.} 
		\label{FDoFandCDoF}
	\end{figure*}
	
	\begin{lemma}
		Assume that $g$ is square-integrable on $\mathcal{A} \subset \mathbb{R}^n$, the kernel function of integral operator $P$ is square-integrable on $\mathcal{A}*\mathcal{B} \subset \mathbb{R}^{n+m}$.
		We define $\left\|  g \right\|:=\sqrt{\int_{\mathcal{A}}|g({\bf x})|^2{\rm d}^n{\bf x}}$ and $\left\| P \right\|:=\sqrt{\int_{\mathcal{B}} \int_{\mathcal{A}}\Big|P({\bf x},{\bf x}')\Big|^2 {\rm d}^n{\bf x}' {\rm d}^m{\bf x} }$. Then we have $\left\| Pg  \right\| \leqslant \left\| P  \right\|\left\| g  \right\| $.
	\end{lemma}
	\begin{proof}
		For $\left\| Pg  \right\|$ we have
		\begin{equation}
			\begin{aligned}
				\left\|  Pg \right\| &= \sqrt{\int_{\mathcal{B}} \Big|\int_{\mathcal{A}}P({\bf x},{\bf x}')g({\bf x}'){\rm d}^n{\bf x}'\Big|^2{\rm d}^m{\bf x}}
				\\&\leqslant \sqrt{\int_{\mathcal{B}} \int_{\mathcal{A}}\Big|P({\bf x},{\bf x}')\Big|^2 {\rm d}^n{\bf x}' {\rm d}^m{\bf x} }
				\sqrt{\int_{\mathcal{A}}\Big|g({\bf x}')\Big|^2{\rm d}^n{\bf x}'},
			\end{aligned}
		\end{equation}
		according to Cauchy–Schwarz inequality.
	\end{proof}
	
	To intuitively show the definition of the functional DoF and channel DoF, we plot Fig. \ref{FDoFandCDoF} which includes the functional space of transmitted and received fields $\mathcal{H}_t$ and $\mathcal{H}_r$. It is shown that the functional DoF corresponds to the number of orthogonal base functions in the functional space needed to construct the space with required accuracy. Moreover, the channel DoF corresponds to the number of subchannels needed to project arbitrary function in $\mathcal{H}_t$ to a function in $\mathcal{H}_r$ with required accuracy. It is expected that the channel DoF is upper-bounded by the functional DoF of the fields at the transceivers. We provide the following theorem to rigorously prove their relationship. 
	
	\begin{theorem}
		(The relationship between the functional DoF and the channel DoF) Let $\mathcal{H}_{\rm t}$ and $\mathcal{H}_{\rm r}$ be two Hilbert spaces and $P:\mathcal{H}_{\rm t} \rightarrow \mathcal{H}_{\rm r}$ be an operator. The functional DoF of $\mathcal{H}_{\rm t}$ with error $\epsilon_1$ is $N_1$. The functional DoF of $\mathcal{H}_{\rm r}$ with error $\epsilon_2$ is $N_2$. Then the channel DoF of the operator $P: \mathcal{H}_{\rm t} \rightarrow \mathcal{H}_{\rm r}$ with error $\left\| P \right\| \epsilon_1 + \left\|  P \right\| (1+\epsilon_1) \epsilon_2$ is at most ${\rm max}\{N_1,N_2\}$, i.e.,
		\begin{equation}
			\begin{aligned}
			&\mathsf{cDoF}_{\left\| P \right\| \epsilon_1 + \left\|  P \right\| (1+\epsilon_1) \epsilon_2}({\mathcal{H}_{\rm t}},{\mathcal{H}_{\rm r}},P) \leqslant {\rm max}\{\mathsf{fDoF}_{\epsilon_1}(\mathcal{H}_{\rm t}),\mathsf{fDoF}_{\epsilon_2}(\mathcal{H}_{\rm r})\}.
			\end{aligned}
		\end{equation}
	\end{theorem}
	
	\begin{proof}
		For the space $\mathcal{H}_{\rm t}$, it has DoF $N_1$ and with error $\epsilon_1$ according to the above definition. The basis on $\mathcal{H}_{\rm t}$ are defined as $g_{i=1}^{N_1}$, which constructs the subspace $\hat{\mathcal{H}}_{\rm t}$. For the space $\mathcal{H}_{\rm r}$, it has DoF $N_2$ with error $\epsilon_2$ according to the above definition. The basis on $\mathcal{H}_{\rm r}$ are defined as $f_{i=1}^{N_2}$, which constructs the subspace $\hat{\mathcal{H}}_{\rm r}$. If $N_1=N_2$, we have a set of operators $P_1^n$ that satisfies $P_ig_j=f_i {\mathbbm 1}_{i=j}$. Moreover, $\bra{P_i g}P_j g\rangle =0$ for arbitrary $g \in \hat{\mathcal{H}}_{\rm t}$ and $i \neq j$. Otherwise, if $N_1<N_2$, we add $N_2-N_1$ base functions $g_{i=N_1+1}^{N_2}$ in $\mathcal{H}_{\rm t}$, which expands $\hat{\mathcal{H}}_t$ to $\tilde{\mathcal{H}}_{\rm t}$. It is obvious that $\mathcal{H}_{\rm t}$ has DoF $N_2$ with error $\tilde{\epsilon}_1\leqslant \epsilon_1$. If $N_1>N_2$, similar operation can be done on $\mathcal{H}_{\rm r}$, showing that $\mathcal{H}_{\rm r}$ has DoF $N_1$ with error $\tilde{\epsilon}_2\leqslant \epsilon_2$.
		
		From the above analysis, we define $N_3 = {\rm max}\{N_1,N_2\}$.
		Now for arbitrary $g \in \mathcal{H}_{\rm t} $ and $\left\| g\right\|=1$, there exists a set of parameters $b_i$ that satisfies $\left\|g-\sum_{i=1}^{N_3}b_ig_i \right\| \leqslant \epsilon_1$. From \eqref{eqn_min_f_approx} we know that the minimum of $\left\|g-\sum_{i=1}^{N_3}b_ig_i\right\|$ is achieved when $b_i = \frac{ \langle g \ket{g_i}}{\left\| g_i \right\|^2}$. Then we have
		\begin{equation}
			\begin{aligned}
				\langle g-\sum_{i=1}^{N_3}b_ig_i \ket{g_j} =\langle g\ket{g_j} - b_i \langle g_i\ket{g_j}=0.
			\end{aligned}
		\end{equation}
 Then we choose $b_i$ that satisfies the above conditions. For arbitrary $P$ we have $\left\| Pg-P \sum_{i=1}^{N_3}b_ig_i \right\| \leqslant \left\| P \right\| \epsilon_1$ according to {\bf Lemma 1}. Noted that there exists $f=P\sum_{i=1}^{N_3}b_ig_i \in \mathcal{H}_{\rm r}$. Moreover, according to the definition of the DoF of space $\mathcal{H}_{\rm r}$, we have $a_{i=1}^{N_3}$ that satisfy $\left\| f-\sum_{i=1}^{N_3} a_i f_i  \right\| \leqslant \left\| f \right\| \tilde{\epsilon}_2 \leqslant \left\|  P \right\| (1+\tilde{\epsilon}_1) \tilde{\epsilon}_2$. For the rank-1 operators $P_i$ we construct them as $P_i = \frac{\ket{f_i} \bra{g_i}}{\left\| g_i^2\right\|^2}$. Then we can construct $P_1 = \sum_{i=1}^{N_3} a_i P_i/b_i$, which leads to $P_1\sum_{i=1}^{N_3}b_ig_i =\sum_{i=1}^{N_3} a_i f_i $.
		
		Based on the construction scheme of $P_1$, we have 
		\begin{equation}
			\begin{aligned}
				P_1 (g-\sum_{i=1}^{N_3}b_ig_i ) &= \sum_{j=1}^{N_3}\frac{a_j}{b_j} \ket{f_j} \bra{g_j} g-\sum_{i=1}^{N_3}b_ig_i \rangle
				 = \sum_{j=1}^{N_3}\frac{a_j}{b_j} \ket{f_j} \left(\bra{g-\sum_{i=1}^{N_3}b_ig_i} g_j \rangle \right)^* = 0.
			\end{aligned}
		\end{equation}
		
		Then, we have the following inequality
		\begin{equation}
			\begin{aligned}
				\left\|  Pg - P_1 g  \right\| &\leqslant \left\| (P-P_1)(g-\sum_{i=1}^{N_3}b_ig_i)  \right\|  + \left\| P\sum_{i=1}^{N_3}b_ig_i -P_1\sum_{i=1}^{N_3}b_ig_i  \right\|
				\\&= \left\| P(g-\sum_{i=1}^{N_3}b_ig_i)  \right\|  + \left\| f - \sum_{i=1}^{N_3} a_i f_i \right\|
				\\& \leqslant \left\| P \right\| \tilde{\epsilon}_1 + \left\|  P \right\| (1+\tilde{\epsilon}_1) \tilde{\epsilon}_2
				\\&\leqslant \left\| P \right\| \epsilon_1 + \left\|  P \right\| (1+\epsilon_1) \epsilon_2.
			\end{aligned}
		\end{equation}
		
			Therefore, we can conclude that the channel DoF is no more than ${\rm max}\{N_1,N_2\}$ with error $\left\| P \right\| \epsilon_1 + \left\|  P \right\| (1+\epsilon_1) \epsilon_2$.
	\end{proof}
	
	\begin{remark}
		Through {\bf Theorem 2} we provide the relationship between functional DoF and channel DoF, showing that the channel DoF is upper bounded by the functional DoF at the transceivers. Here it is worth noting that the bound is a weak bound while a tight bound related to ${\rm min}\{N_1,N_2\}$ is expected inspired by classical discrete MIMO theory, which remains for further works.
	\end{remark}
	
	\section{Numerical simulation}
	
	After showing that the DoF analyzed in this paper provides performance upperbound for wireless communication systems, in this part we will provide numerical simulations of the eigenvalues and eigenfunctions based on the theoretical derivations in {\bf Section. \ref{sec_theoretical_dof}}. 
	
	Firstly we simulate the eigenvalues of the Slepian concentration problem for three-dimensional antenna array based on (\ref{eigenfunction_3d}), where the kernel functions are provided in (\ref{equ_ball}), (\ref{equ_spherical_shell}), and (\ref{equ_sphere_surface}). The simulation results show the DoF upperbound when antennas are constrained in three-dimensional space.
	We show the DoF with different frequency bandwidth in Fig. \ref{different_frequency_setting}. The three-dimensional antenna array is $0.5\,{\rm m} \times 0.5\,{\rm m} \times 0.5\,{\rm m}$, sampled with $\lambda/2$ antenna spacing. The shape of the concentration region in wavenumber domain is set to ball, spherical shell, and sphere surface, which correspond to different frequency settings.
	The maximum frequency is $3\,{\rm GHz}$. The ball corresponds to frequency band $[0,3]\,{\rm GHz}$. The thick shell has $2\,{\rm GHz}$ bandwidth. The thin shell has $150\,{\rm MHz}$ bandwidth. The sphere surface in the wavenumber domain corresponds to single-frequency point at $3\,{\rm GHz}$. The figure shows that not all eigenvalues tend to maximum or 0 for broad-band scenario. On the contrary, a large transition band still exist, which means that the asymptotic conclusions of DoF are not accurate enough. For narrow-band or single frequency scenario, the eigenvalues decay very fast from the maximum value, which implies that asymptotic conclusions of DoF do not work.
	It can also be observed that when the frequency bandwidth tends to 0 (thickness of the spherical shell in wavenumber domain tends to 0), the eigenvalues converge to that in the case with single frequency point, which leads to non-zeros DoFs and verifies the theoretical analysis. It also shows that the DoF is not linear with bandwidth in narrow-band scenario.

	\begin{figure}
		\centering 
		\includegraphics[width=0.7\textwidth]{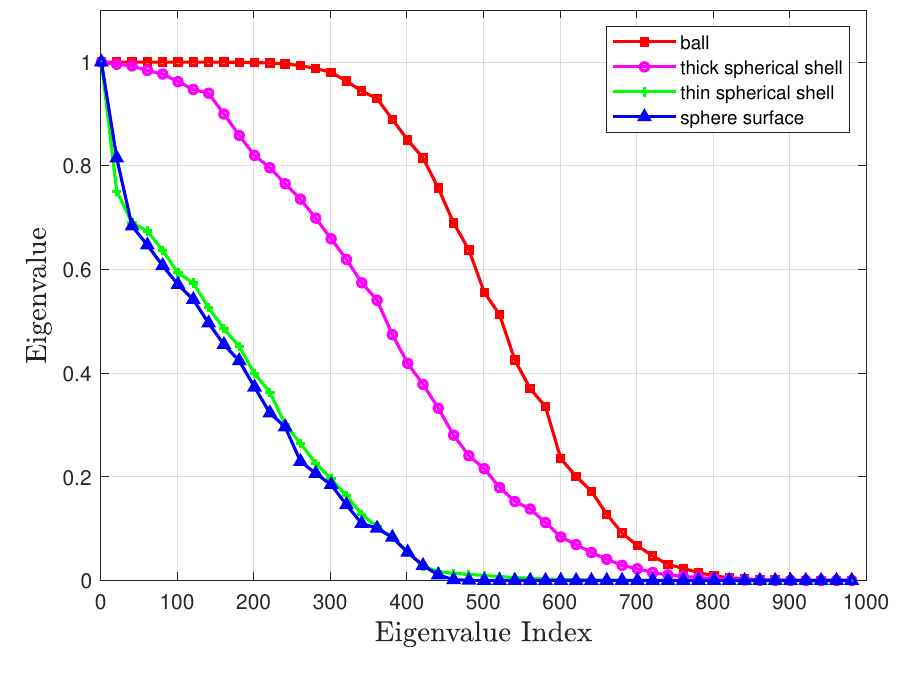} 
		\caption{The DoF gain with different frequency band. The three-dimensional antenna array is $0.5\,{\rm m} \times 0.5\,{\rm m} \times 0.5\,{\rm m}$, sampled with $\lambda/2$ antenna spacing. The maximum frequency is $3\,{\rm GHz}$. The thick shell has $2\,{\rm GHz}$ bandwidth. The thin shell has $150\,{\rm MHz}$ bandwidth. The sphere surface in the wavenumber domain corresponds to single frequency point at $3\,{\rm GHz}$.} 
		\label{different_frequency_setting}
	\end{figure}

	Then, we simulate the DoF gain with different heights in Fig. \ref{different_z_setting}. The three-dimensional antenna array is $0.5\,{\rm m} \times 0.5\,{\rm m} \times 0.05x\,{\rm m}$, sampled with $\lambda/2$ antenna spacing. $x$ is chosen to be $1$, $3$, $7$, and $10$ to represent three-dimensional antenna array with different heights (number of layers). It can be observed from Fig. \ref{different_z_setting} that extra space dimension brings DoF gain for the  electromagnetic fields, which means that by using three-dimensional antenna arrays, more DoFs may be obtained compared to traditional two-dimensional surface. 
	
	\begin{remark}
		The performance gain of three-dimensional antenna array may come from the enhancement of wave number domain resolution capability compared to two-dimensional antenna array. For the three-dimensional concentration problem with wavenumber domain constraint $k_x^2+k_y^2+k_z^2 \leqslant k_0^2$, it is easy to observe that it has the same asymptotic DoFs with the problem which has wavenumber domain constraint $k_x^2+k_y^2 \leqslant k_0^2$ and $|k_z| \leqslant \frac{2}{3}k_0$. The latter problem, which has a cylinder-shaped wavenumber domain constraint, can use separation of variables. To be clearer, it can be decomposed as a two-dimensional concentration problem with variables $k_x$ and $k_y$, which has eigenfunctions $\{\phi_i(k_x,k_y)\}_{i=1}^{+\infty}$, and a one-dimensional concentration problem with variables $k_z$, which has eigenfunctions $\{\phi_i(k_z)\}_{i=1}^{+\infty}$. Then we can observe that the eigenfunctions of the three-dimensional concentration problem, expressed as $\{\phi_i(k_x,k_y)\phi_j(k_z)\}_{i=1,j=1}^{+\infty}$, from a many-to-one mapping with the eigenfunctions of the two-dimensional concentration problem. In another word, it means that several waveforms in three-dimensional case project to the same waveform in two-dimensional case, showing the enhancement of wave number domain resolution capability by using three-dimensional array.
	\end{remark}
	
	\begin{figure}
		\centering 
		\includegraphics[width=0.7\textwidth]{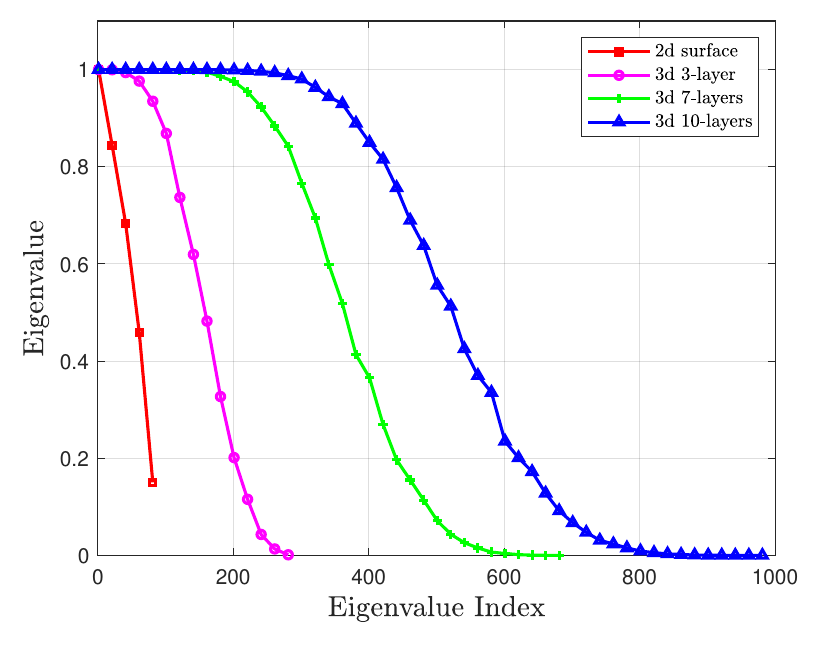} 
		\caption{The DoF gain with different heights. The three-dimensional antenna array is $0.5\,{\rm m} \times 0.5\,{\rm m} \times 0.05x\,{\rm m}$, sampled with $\lambda/2$ antenna spacing. $x$ is chosen to be $1$, $3$, $7$ and $10$ to represent three-dimensional antenna array with different heights (number of layers).} 
		\label{different_z_setting}
	\end{figure}

	Next we discuss the near-optimal sampling scheme of antenna array. Noted that for the two-dimensional antenna array that simply omit the third spatial dimension, \cite{pizzo2022nyquist} shows that the half-wavelength sampling is the optimal scheme. On the one hand, this conclusion is approximate rather than completely accurate even for an ideal two-dimensional array. On the other hand, electromagnetic fields actually exist in three-dimensional space, which implies that simply ignoring one dimension may lead to model mismatch. In the following part we will show the spatial sampling scheme of ideal two-dimensional planar array based on the two-dimensional Slepian concentration problem. Then we will provide simulations based on three-dimensional models to provide more accurate and general conclusions on DoF.
		
		Fig. \ref{different_z_setting_2d} shows the eigenvalues of a two-dimensional Slepian concentration problem, which can be viewed as omitting the third dimension in space and wavenumber domain of electromagnetic field. The space concentration region is a rectangle, the size of which is $6\lambda \times 6\lambda$, with different antenna spacing in the simulation. The maximum frequency is $3\,{\rm GHz}$. It is worth noting that whatever the frequency bandwidth is, the concentration region in the space domain keeps the same, because the two-dimensional projections of solid sphere, spherical shell and sphere surface are all solid circles. 
		 It can be observed from Fig. \ref{different_z_setting_2d} that by applying half-wavelength sampling on the two-dimensional array, nearly all eigenvalues are above zero, which means that half-wavelength sampling is necessary to utilize the DoF of the electromagnetic fields. The number of non-zero eigenvalues of $\lambda/4$ sampling is a little larger than half-wavelength sampling, which introduces extra DoFs. Whether these DoFs are worth using relies on the error threshold in {\bf Definition 1}.
		 
		 \begin{figure}
		 	\centering 
		 	\includegraphics[width=0.7\textwidth]{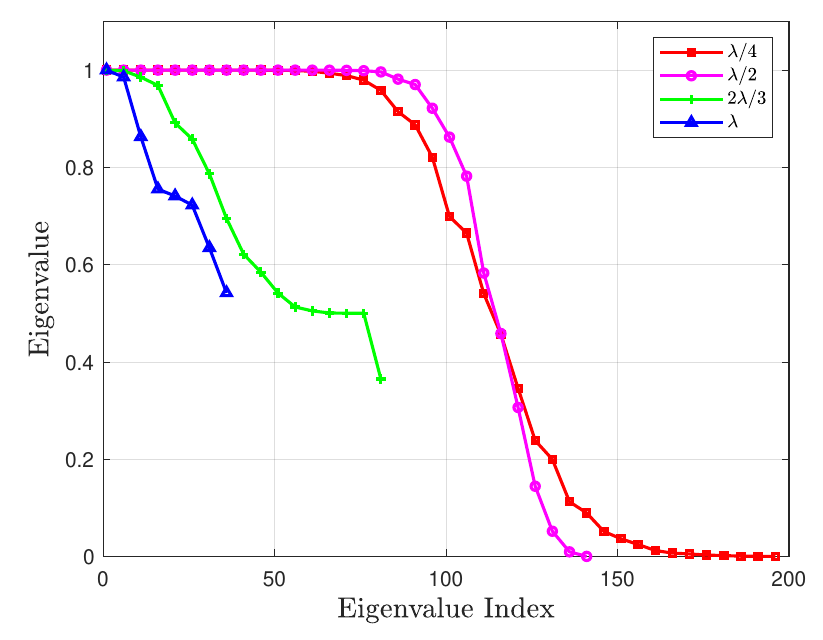} 
		 	\caption{The DoF gain from the ideal two-dimensional Slepian's problem. The antenna array is $6\lambda \times 6\lambda$ with different antenna spacing on $x$ and $y$ axes. The frequency is $3\,{\rm GHz}$.} 
		 	\label{different_z_setting_2d}
		 \end{figure}
		 
		 Figure \ref{different_antenna_spacing_merge} depicts the eigenvalues corresponding to the three-dimensional Slepian concentration problem, where no dimensions of the electromagnetic field's spatial and wavenumber domains are neglected. The size of spatial concentration region is $6\lambda \times 6\lambda \times 2.5\lambda$. To facilitate comparison with the two-dimensional case, the sampling spacing is varied only along the $x$ and $y$ axes, while maintaining half-wavelength sampling along the $z$ axis. It is worth noting that the two-dimensional projection of this three-dimensional array on the $xOy$ plane is consistent with the scenario in Figure \ref{different_z_setting_2d}. According to the conclusions from the previous figure, the characteristics of different sampling spacing on the $x$ and $y$ axes should only be related to the highest frequency and not to the bandwidth, achieving near-optimal results at the half-wavelength sampling. In figure \ref{different_antenna_spacing_merge} we have two different scenarios: broad-band and single-frequency point, both with a highest frequency of $3\,{\rm GHz}$. The broad-band scenario covers $[0,3]\,{\rm GHz}$, corresponding to a solid sphere in the wavenumber concentration region, while the single-frequency point corresponds to a sphere surface in the wavenumber concentration region. It is shown that the characteristics of different sampling spacing differ between the two scenarios. For the broad-band scenario, half-wavelength spacing is necessary, similar to the simulation results of the two-dimensional model. However, for the single-frequency point scenario, about one-third of the eigenvalues tend to zero at half-wavelength spacing, which means that half-wavelength spacing is oversampling. On the contrary, at $2\lambda/3$ spacing, almost all eigenvalues are significantly greater than zero, suggesting that the near-optimal sampling density is around $2\lambda/3$ spacing, which is sparser than half-wavelength sampling. Therefore, for electromagnetic fields in three-dimensional space and wavenumber domain, the near-optimal sampling spacing depends on the structure of the frequency and wavenumber concentration regions, not just the highest frequency. When only considering the two-dimensional problem, the structure of the frequency and wavenumber concentration regions is lost, leading to less accurate results.

\begin{figure}
	\centering 
	\includegraphics[width=0.7\textwidth]{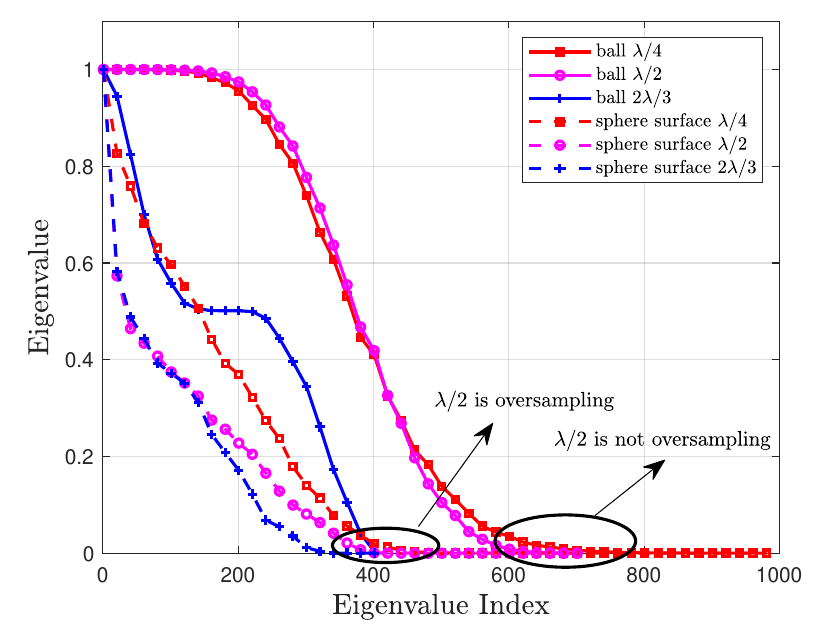} 
	\caption{The DoF gain with different antenna spacing. The 3d-antenna array is $6\lambda \times 6\lambda \times 2.5\lambda$ with different antenna spacing on $x$ and $y$ axes, and $\lambda/2$ antenna spacing on $z$ axis. The frequency is $3\,{\rm GHz}$.} 
	\label{different_antenna_spacing_merge}
\end{figure}

		Moreover, we will discuss the influence of spatial sampling scheme, which corresponds to the spacing between different antennas. In Fig. \ref{different_sampling_scheme_merge} we use both equal antenna spacing and unequal antenna spacing based on Gauss-Legendre quadrature rule. The sampling points and weights of Gauss-Legendre quadrature rule rely on the following Legendre polynomials
		$P_0(x) :=  1$ and
		$P_n(x) := \frac{1}{2^n n!}\frac{{\rm d}^n}{{\rm d}x^n}\left[ (x^2-1)^n \right]$. The sampling points $x_1,\dots,x_n$ are zeros of $P_{n}(x)$, and the weights $A_1,\cdots,A_n$ satisfy
		\begin{equation}
			\begin{aligned}
				A_k = \frac{2}{n} \frac{1}{P_{n-1}(x_k)P^{'}_{n}(x_k)}.
			\end{aligned}
		\end{equation}
		It is obvious that unequal antenna spacing does not change the number of non-zero eigenvalues, which means that the DoF will not be greatly influenced and converge to the same upperbound determined by continuous aperture. However, unequal antenna spacing influences the decay rate of eigenvalues. In our simulation results, Gauss-Legendre sampling scheme corresponds to more large eigenvalues compared to classical equal sampling scheme. This phenomenon may be used when we need to design limited number of electromagnetic patterns to convey information. With Gauss-Legendre sampling scheme we can have a better approximation to expected electromagnetic fields with a fixed number of patterns, and use fewer patterns for given approximation accuracy to expected electromagnetic fields.

\begin{figure}
	\centering 
	\includegraphics[width=0.7\textwidth]{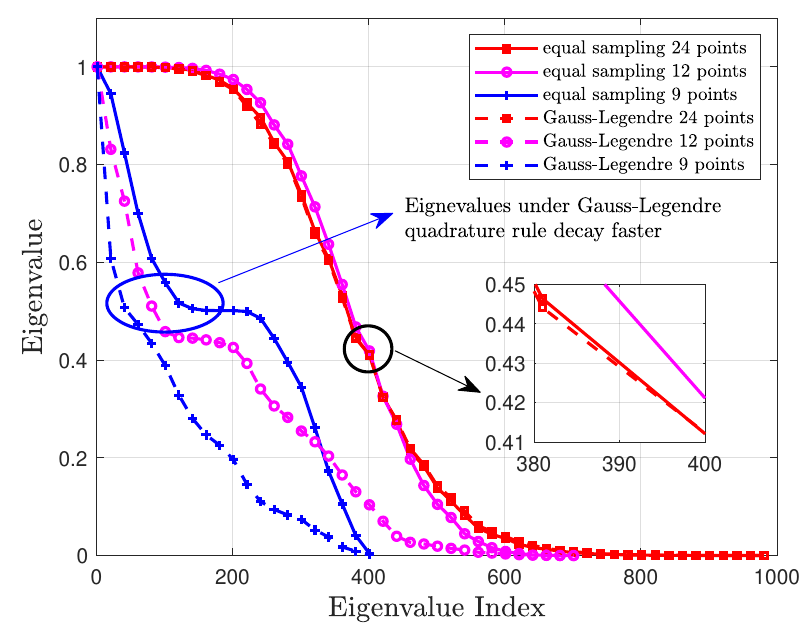} 
	\caption{The DoF gain with different sampling points with equal sampling and Gauss-Legendre quadrature rule. The 3d-antenna array is $6\lambda \times 6\lambda \times 2.5\lambda$ with different sampling points on $x$ and $y$ axes, and $\lambda/2$ antenna spacing on $z$ axis. The frequency band is $[0,3]\,{\rm GHz}$.} 
	\label{different_sampling_scheme_merge}
\end{figure}

Finally we consider both the space and time domains to provide more general results for electromagnetic information theory. Noted that for wireless communication systems the time-harmonic assumption requires the periodicity of the electromagnetic field in the time domain, that is, the theoretical model is strictly satisfied only after the electromagnetic field becomes stationary. The above communication process will greatly reduce the efficiency of wireless communication and is not practical enough. Therefore, without the time-harmonic assumption, we can obtain a more general model for electromagnetic information theory, which better fits practical wireless communication processes. With the theoretical analysis, we plot the space-time patterns inspired by Slepian's concentration problem in Fig. \ref{fig_cdl_fit}. For ease of calculation, we bound the electromagnetic fields on one-dimensional space domain which corresponds to linear antenna array. These patterns have good orthogonality, shown in Fig. \ref{pattern_correlation}. Moreover, according to \eqref{eqn_min_f_approx} we know that these patterns can be used to construct arbitrary required electromagnetic fields in the source/destination regions with a given approximation error.

	\begin{figure}[!t]
	\setlength{\abovecaptionskip}{-0.0cm}
	\setlength{\belowcaptionskip}{-0.0cm}
	\centering
	\subfigcapskip -0.5em
	\subfigure[Pattern 1]{
		\includegraphics[width=3in]{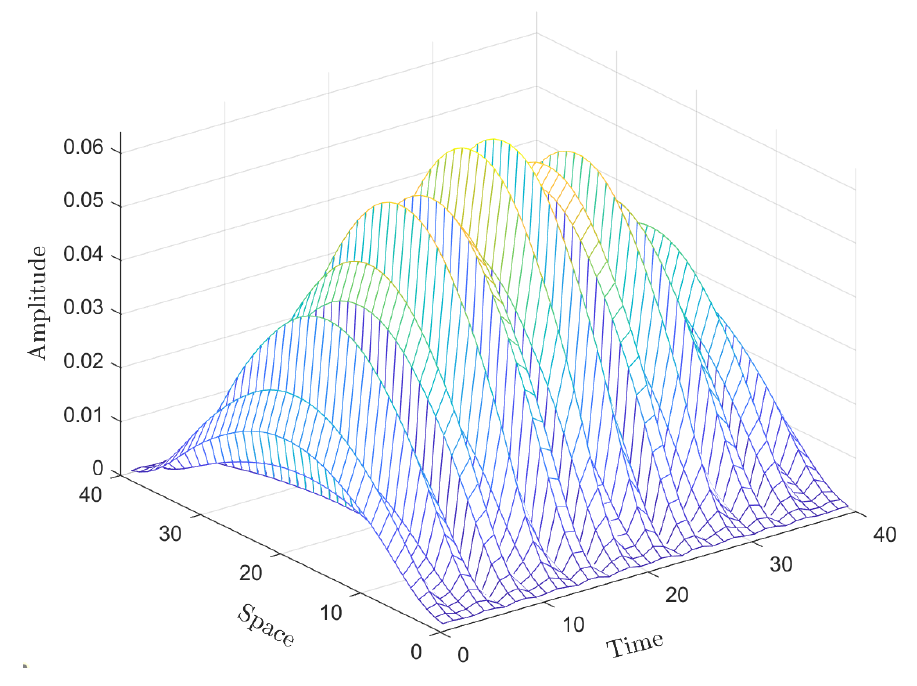}
	}%\hspace{-2em}
	\subfigure[Pattern 2]{
		\includegraphics[width=3in]{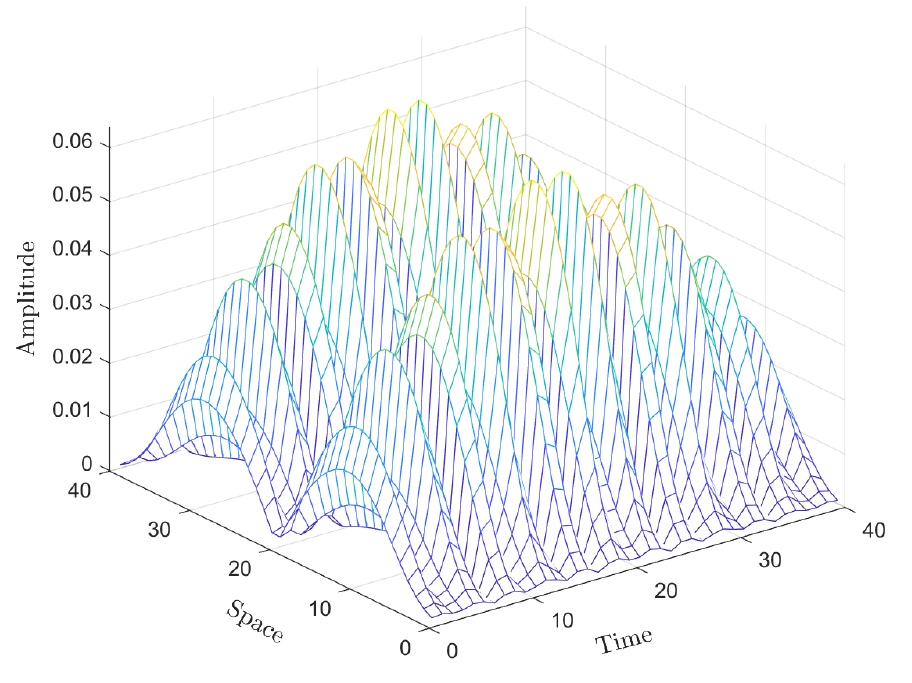}
	}
	
	\subfigure[Pattern 3]{
		\includegraphics[width=3in]{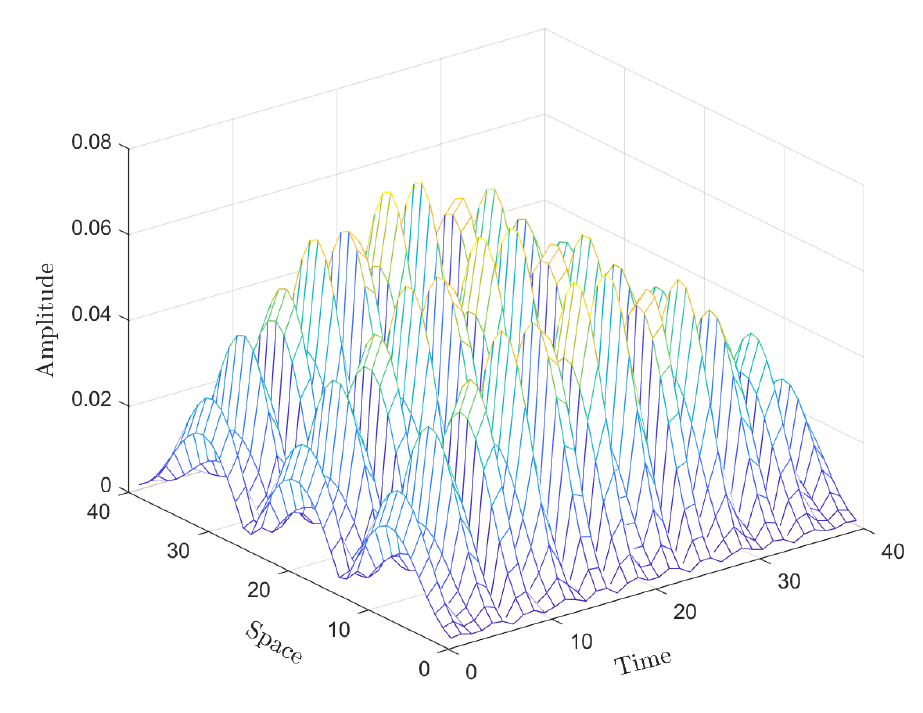}
	}%\hspace{-2em}
	\subfigure[Pattern 4]{
		\includegraphics[width=3in]{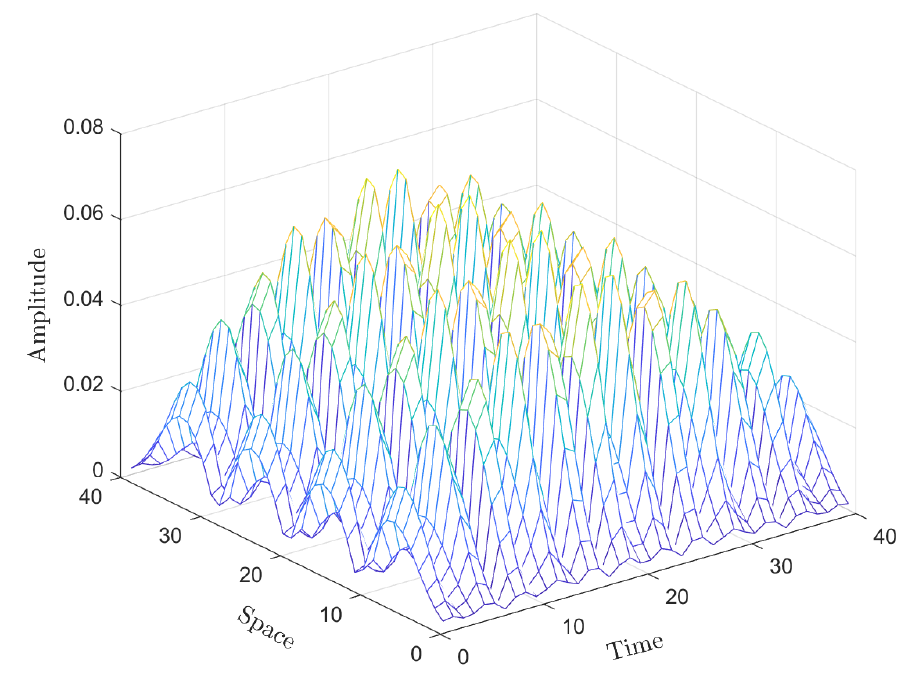}
	}
	
	\subfigure[Pattern 5]{
		\includegraphics[width=3in]{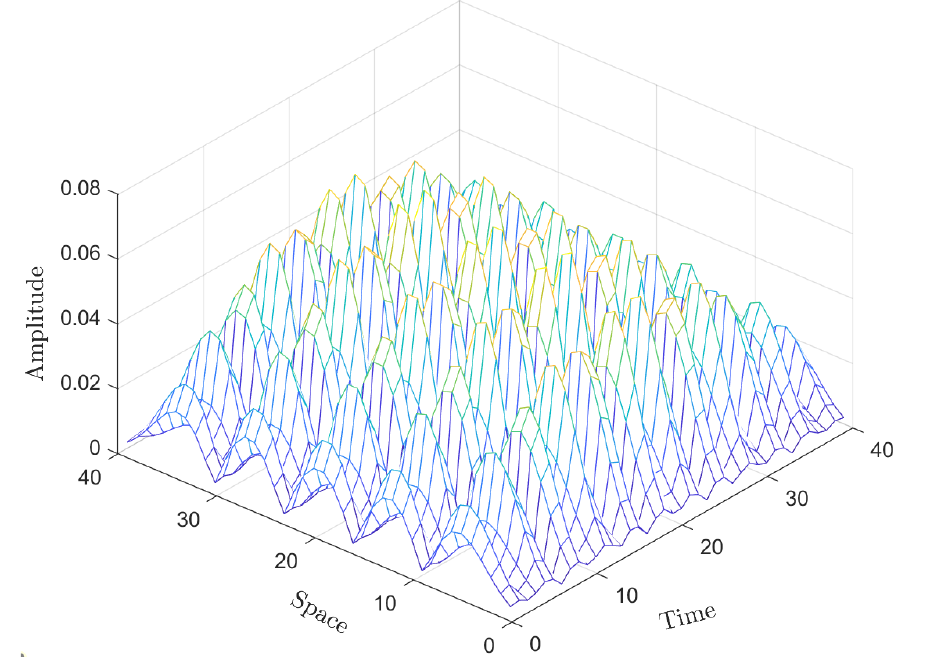}
	}%\hspace{-2em}
	\subfigure[Pattern 6]{
		\includegraphics[width=3in]{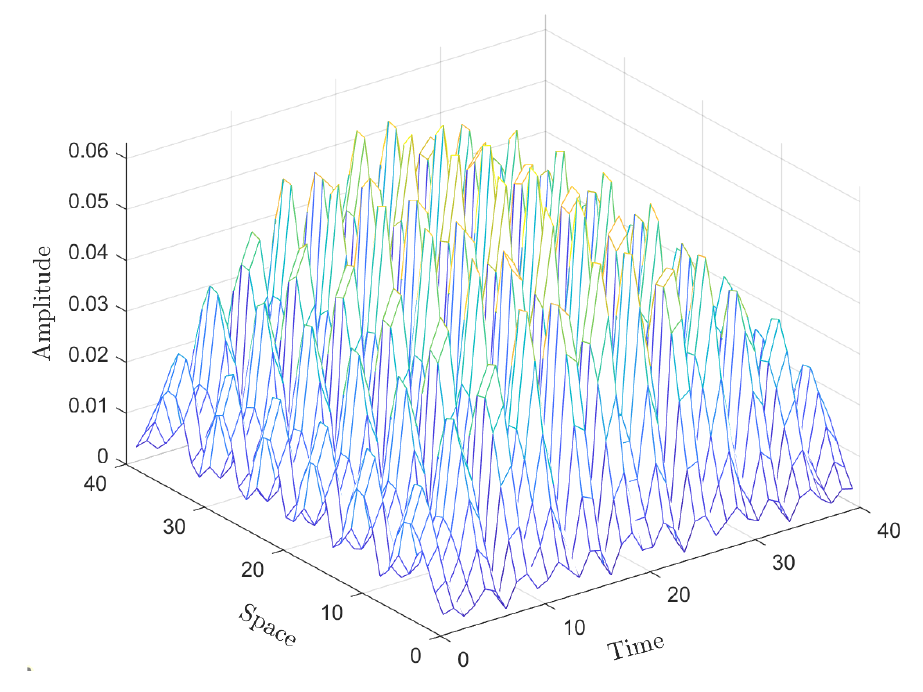}
	}
	
	%\hspace{-2em}
	%\vspace{1em}
	\caption{
		The first 6 orthogonal space-time patterns that satisfy the electromagnetic constraints and concentrate in the observation region. 
	}
	%\vspace{-1em}
	\label{fig_cdl_fit}
\end{figure}

\begin{figure}
	\centering 
	\includegraphics[width=0.7\textwidth]{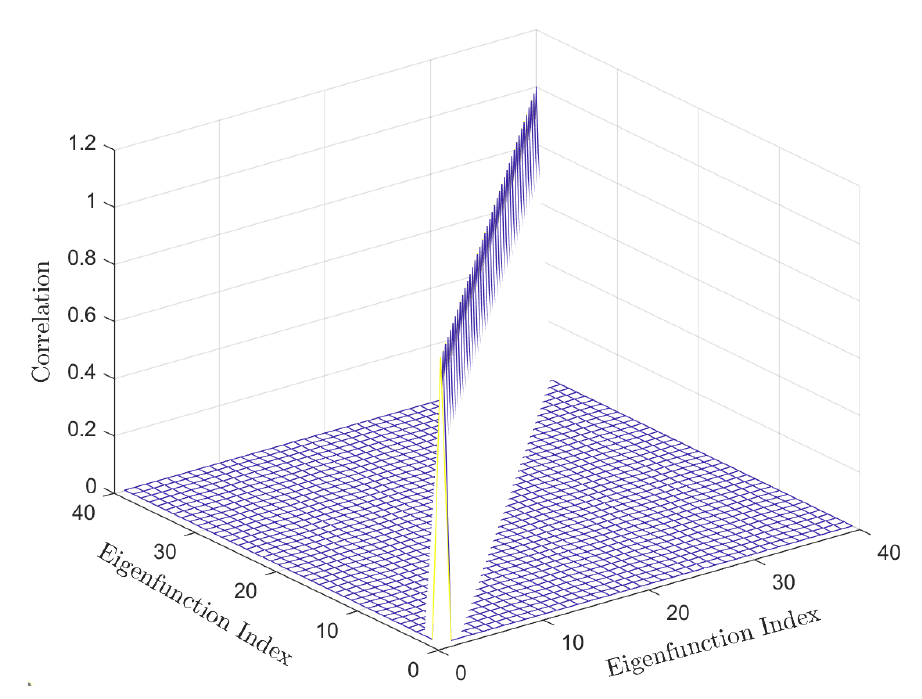} 
	\caption{The correlation of the space-time patterns.} 
	\label{pattern_correlation}
\end{figure}

\section{Conclusion}
In this paper, we provide a general DoF and pattern analyzing scheme for EIT. Specifically, we begin from the classical Slepian concentration problem which analyzes the model with limited time observation region and frequency bandwidth. Then we extend it to three-dimensional space domain and four-dimensional space-time domain to obtain the DoF of the electromagnetic fields with given physical constraints. We provide mathematical definitions of functional DoF and channel DoF and theoretically discuss the relationship between them. Finally we provide numerical simulations to verify the theoretical analysis and show more insights, including the near-optimal spatial sampling interval of antenna array, the DoF of three-dimensional antenna array, the impact of unequal antenna spacing, the orthogonal space-time patterns, etc. For example, our simulation results show that with narrow frequency bandwidth, half-wavelength antenna spacing may not be necessary to achieve the DoF upperbound.

Further research work can be focused on deriving analytical expressions of eigenvalues and eigenfunctions instead of numerical simulations in some specific cases. Moreover, a tighter bound between functional DoF and channel DoF is expected.

\section*{Acknowledgment}

The authors sincerely thank Prof. Lizhong Zheng and Prof. Husheng Li for their valuable comments on this work.

	\footnotesize
	
	\bibliographystyle{IEEEtran}
	%\clearpage
	
	\bibliography{bib}

\begin{thebibliography}{10}
\providecommand{\url}[1]{#1}
\csname url@samestyle\endcsname
\providecommand{\newblock}{\relax}
\providecommand{\bibinfo}[2]{#2}
\providecommand{\BIBentrySTDinterwordspacing}{\spaceskip=0pt\relax}
\providecommand{\BIBentryALTinterwordstretchfactor}{4}
\providecommand{\BIBentryALTinterwordspacing}{\spaceskip=\fontdimen2\font plus
\BIBentryALTinterwordstretchfactor\fontdimen3\font minus
  \fontdimen4\font\relax}
\providecommand{\BIBforeignlanguage}[2]{{%
\expandafter\ifx\csname l@#1\endcsname\relax
\typeout{** WARNING: IEEEtran.bst: No hyphenation pattern has been}%
\typeout{** loaded for the language `#1'. Using the pattern for}%
\typeout{** the default language instead.}%
\else
\language=\csname l@#1\endcsname
\fi
#2}}
\providecommand{\BIBdecl}{\relax}
\BIBdecl

\bibitem{na2024operator}
M.~Na, J.~Lee, G.~Choi, T.~Yu, J.~Choi, J.~Lee, and S.~Bahk, ``Operator's
  perspective on {6G}: {6G} services, vision, and spectrum,'' \emph{IEEE
  Commun. Mag.}, vol.~62, no.~8, pp. 178--184, Aug. 2024.

\bibitem{basar2019wireless}
E.~Basar, M.~Di~Renzo, J.~De~Rosny, M.~Debbah, M.-S. Alouini, and R.~Zhang,
  ``Wireless communications through reconfigurable intelligent surfaces,''
  \emph{IEEE Access}, vol.~7, pp. 116\,753--116\,773, Aug. 2019.

\bibitem{wang2022location}
Z.~Wang, Z.~Liu, Y.~Shen, A.~Conti, and M.~Z. Win, ``Location awareness in
  beyond {5G} networks via reconfigurable intelligent surfaces,'' \emph{{IEEE}
  J. Sel. Areas Commun.}, vol.~40, no.~7, pp. 2011--2025, Jul. 2022.

\bibitem{huang2020holographic}
C.~Huang, S.~Hu, G.~C. Alexandropoulos, A.~Zappone, C.~Yuen, R.~Zhang,
  M.~Di~Renzo, and M.~Debbah, ``Holographic {MIMO} surfaces for {6G} wireless
  networks: Opportunities, challenges, and trends,'' \emph{IEEE Wireless
  Commun.}, vol.~27, no.~5, pp. 118--125, Oct. 2020.

\bibitem{zhang2023pattern}
Z.~Zhang and L.~Dai, ``Pattern-division multiplexing for multi-user
  continuous-aperture {MIMO},'' \emph{IEEE J. Sel. Areas Commun.}, vol.~41,
  no.~8, pp. 2350--2366, Aug. 2023.

\bibitem{cui2022near}
M.~Cui, Z.~Wu, Y.~Lu, X.~Wei, and L.~Dai, ``Near-field {MIMO} communications
  for {6G}: Fundamentals, challenges, potentials, and future directions,''
  \emph{IEEE Comm. Mag.}, vol.~61, no.~1, pp. 40--46, Jan. 2023.

\bibitem{wu2023multiple}
Z.~Wu and L.~Dai, ``Multiple access for near-field communications: {SDMA} or
  {LDMA}?'' \emph{{IEEE} J. Sel. Areas Commun.}, vol.~41, no.~6, pp.
  1918--1935, Jun. 2023.

\bibitem{chafii2023twelve}
M.~Chafii, L.~Bariah, S.~Muhaidat, and M.~Debbah, ``Twelve scientific
  challenges for {6G}: Rethinking the foundations of communications theory,''
  \emph{IEEE Comm. Surveys Tut.}, Feb. 2023.

\bibitem{migliore2018horse}
M.~D. Migliore, ``Horse (electromagnetics) is more important than horseman
  (information) for wireless transmission,'' \emph{{IEEE} Trans. Antennas
  Propag.}, vol.~67, no.~4, pp. 2046--2055, Apr. 2018.

\bibitem{zhu2022electromagnetic}
J.~Zhu, Z.~Wan, L.~Dai, M.~Debbah, and H.~V. Poor, ``Electromagnetic
  information theory: Fundamentals, modeling, applications, and open
  problems,'' \emph{{IEEE} Wireless Commun.}, early access, Jan. 2024, doi:
  10.1109/MWC.019.2200602.

\bibitem{gong2023holographic}
T.~Gong, L.~Wei, C.~Huang, Z.~Yang, J.~He, M.~Debbah, and C.~Yuen,
  ``Holographic {MIMO} communications with arbitrary surface placements:
  Near-field {LoS} channel model and capacity limit,'' \emph{IEEE J. Sel. Areas
  Commun.}, early access, Apr. 2023, doi: 10.1109/JSAC.2024.3389126.

\bibitem{wei2023tri}
L.~Wei, C.~Huang, G.~C. Alexandropoulos, Z.~Yang, J.~Yang, E.~Wei, Z.~Zhang,
  M.~Debbah, and C.~Yuen, ``Tri-polarized holographic {MIMO} surfaces for
  near-field communications: Channel modeling and precoding design,''
  \emph{IEEE Trans. Wireless Commun.}, vol.~22, no.~12, pp. 8828--8842, Dec.
  2023.

\bibitem{pizzo2023wide}
A.~Pizzo, A.~Lozano, S.~Rangan, and T.~L. Marzetta, ``Wide-aperture {MIMO} via
  reflection off a smooth surface,'' \emph{{IEEE} Trans. Wireless Commun.},
  vol.~22, no.~8, pp. 5229--5239, Aug. 2023.

\bibitem{bucci1987spatial}
O.~Bucci and G.~Franceschetti, ``On the spatial bandwidth of scattered
  fields,'' \emph{{IEEE} Trans. Antennas Propag.}, vol.~35, no.~12, pp.
  1445--1455, Dec. 1987.

\bibitem{bucci1989degrees}
O.~M. Bucci and G.~Franceschetti, ``On the degrees of freedom of scattered
  fields,'' \emph{{IEEE} Trans. Antennas Propag.}, vol.~37, no.~7, pp.
  918--926, Jul. 1989.

\bibitem{franceschetti2015landau}
M.~Franceschetti, ``On {Landau’s} eigenvalue theorem and information
  cut-sets,'' \emph{{IEEE} Trans. Inf. Theory}, vol.~61, no.~9, pp. 5042--5051,
  Sep. 2015.

\bibitem{jensen2008capacity}
M.~A. Jensen and J.~W. Wallace, ``Capacity of the continuous-space
  electromagnetic channel,'' \emph{{IEEE} Trans. Antennas Propag.}, vol.~56,
  no.~2, pp. 524--531, Feb. 2008.

\bibitem{jeon2017capacity}
W.~Jeon and S.-Y. Chung, ``Capacity of continuous-space electromagnetic
  channels with lossy transceivers,'' \emph{{IEEE} Trans. Inf. Theory},
  vol.~64, no.~3, pp. 1977--1991, Mar. 2018.

\bibitem{wan2022mutual}
Z.~Wan, J.~Zhu, Z.~Zhang, L.~Dai, and C.-B. Chae, ``Mutual information for
  electromagnetic information theory based on random fields,'' \emph{{IEEE}
  Trans. Commun.}, vol.~71, no.~4, pp. 1982--1996, Feb. 2023.

\bibitem{slepian1976bandwidth}
D.~Slepian, ``On bandwidth,'' \emph{Proc. of the IEEE}, vol.~64, no.~3, pp.
  292--300, Mar. 1976.

\bibitem{balanis2015antenna}
C.~A. Balanis, \emph{Antenna theory: analysis and design}.\hskip 1em plus 0.5em
  minus 0.4em\relax John wiley \& sons, 2015.

\bibitem{landau1980eigenvalue}
H.~J. Landau and H.~Widom, ``Eigenvalue distribution of time and frequency
  limiting,'' \emph{Journal of Mathematical Analysis and Applications},
  vol.~77, no.~2, pp. 469--481, 1980.

\bibitem{pizzo2022nyquist}
A.~Pizzo, A.~de~Jesus~Torres, L.~Sanguinetti, and T.~L. Marzetta, ``Nyquist
  sampling and degrees of freedom of electromagnetic fields,'' \emph{IEEE
  Trans. on Signal Process.}, vol.~70, pp. 3935--3947, Jun. 2022.

\bibitem{slepian1961prolate}
D.~Slepian and H.~O. Pollak, ``Prolate spheroidal wave functions, fourier
  analysis and uncertainty—i,'' \emph{Bell System Technical Journal},
  vol.~40, no.~1, pp. 43--63, Jan. 1961.

\bibitem{beylkin2007grids}
G.~Beylkin, C.~Kurcz, and L.~Monz{\'o}n, ``Grids and transforms for
  band-limited functions in a disk,'' \emph{Inverse Problems}, vol.~23, no.~5,
  p. 2059, 2007.

\bibitem{ihara1993information}
S.~Ihara, \emph{Information theory for continuous systems}.\hskip 1em plus
  0.5em minus 0.4em\relax World Scientific, 1993, vol.~2.

\bibitem{bornemann2010numerical}
F.~Bornemann, ``On the numerical evaluation of {Fredholm} determinants,''
  \emph{Mathematics of Computation}, vol.~79, no. 270, pp. 871--915, Sep. 2009.

\end{thebibliography}
\end{document}